\documentclass[a4paper,aps,prl,showpacs,twocolumn,superscriptaddress]{revtex4-1}   

\usepackage[utf8]{inputenc}  
\usepackage[T1]{fontenc}     
\usepackage[british]{babel}  
\usepackage{lmodern}  
\usepackage[scaled=1.03]{inconsolata} 
\usepackage{xcolor} 
\usepackage[colorlinks]{hyperref}  
\definecolor{darkred}{rgb}{0.5,0,0}
\definecolor{darkgreen}{rgb}{0,0.5,0}
\definecolor{darkblue}{rgb}{0,0,0.5}
\hypersetup{colorlinks,breaklinks,linkcolor=darkred,citecolor=darkgreen,urlcolor=darkblue}
\usepackage{graphicx} 
\usepackage{tikz}
\usepackage[babel]{microtype}  
\usepackage{amsmath,amssymb,amsthm,bm,mathtools,amsfonts,mathrsfs,bbm,dsfont} 
\usepackage{xspace}  
\usepackage{multirow}
\usepackage{verbatim}

\usepackage{physics}
\newcommand{\id}{\ensuremath{\mathds{1}}}
\usepackage{bbold}

\renewcommand{\vec}[1]{\boldsymbol{#1}}
\renewcommand{\leq}{\leqslant}

\renewcommand{\geq}{\geqslant}

\renewcommand{\>}{\rangle}
\newcommand{\hil}{\mathcal{H}}
\renewcommand{\tr}[1]{{\rm tr}\left[#1\right]}

\newcommand{\mc}[1]{\mathcal{#1}}

\newcommand{\N}{\mathbb{N}}
\newcommand{\R}{\mathbb{R}}
\newcommand{\C}{\mathbb{C}}

\newtheoremstyle{mystyle}
{6pt}
{6pt}
{\normalfont}
{0pt}
{\bf}
{.}
{ }
{}

\theoremstyle{mystyle}
\newtheorem{theorem}{Theorem}
\newtheorem{lemma}{Lemma}

\newtheorem{observation}{Observation}

\begin{document}
\nonfrenchspacing
\title{An operational characterization of infinite-dimensional quantum resources}

\author{Erkka Haapasalo}
\email{erkkath@gmail.com}
\affiliation{Centre for Quantum Technologies, National University of Singapore, Science Drive 2 Block S15-03-18, Singapore 117543}

\author{Tristan Kraft}
\affiliation{Naturwissenschaftlich-Technische Fakult\"at, Universit\"at Siegen, Walter-Flex-Str.~3, D-57068 Siegen, Germany}

\author{Juha-Pekka Pellonp\"a\"a}
\affiliation{QTF Centre of Excellence, Turku Centre for Quantum Physics, Department of Physics and Astronomy,
  University of Turku, FI-20014 Turun yliopisto, Finland
}

\author{Roope Uola}
\affiliation{D\'{e}partement de Physique Appliqu\'{e}e, Universit\'{e}  de Gen\`{e}ve, CH-1211 Gen\`{e}ve, Switzerland}

\date{\today}  

\begin{abstract}
 Recently, various non-classical properties of quantum states and channels have been characterized through an advantage they provide in specific quantum information tasks over their classical counterparts. Such advantage can be typically proven to be quantitative, in that larger amounts of quantum resources lead to better performance in the corresponding tasks. So far, these characterizations have been established only in the finite-dimensional setting. In this manuscript, we present a technique for extending the known results to the infinite-dimensional regime. The technique relies on approximating infinite-dimensional resource measures by their finite-dimensional counterparts. We give a sufficient condition for the approximation procedure to be tight, i.e. to match with established infinite-dimensional resource quantifiers, and another sufficient condition for the procedure to match with relevant extensions of these quantifiers. We show that various continuous variable quantum resources fall under these conditions, hence, giving them an operational interpretation through the advantage they can provide in so-called quantum games. Finally, we extend the interpretation to the max relative entropy in the infinite-dimensional setting.
\end{abstract}

\maketitle

\textit{Introduction.---} The idea of unified treatment of the non-classical properties of quantum mechanics has led to the development of quantum resource theories~\cite{ReviewQRT}. Quantum resource theories are built on the notions of free states and free operations. The former are quantum objects that have no resource content, and the latter are maps that do not convert free states into resource states. As an example, in the resource theory of entanglement, separable states are free, and local operations assisted by classical communication are free operations. To quantify the amount of a given resource, quantum resource theories use measures that are faithful, i.e., obtain the value zero exactly on the free set, and do not increase under the action of free operations. The well-known max-relative entropy of entanglement~\cite{Datta2009a,Datta2009b}, the robustness of coherence~\cite{NBC+16}, and the accessible information in quantum-to-classical channels~\cite{SL19} are examples of such measures. Exploiting the convex structure of quantum resource theories various resource quantifiers have been shown to have an interpretation as the advantage a resource can provide over its classical counterparts in some quantum information task~\cite{Piani15,NBC+16,Takagi19,BCP19,TakagiRegula19,SL19,LPS19,DLPS20,Uola20}. For example, the best separable approximation~\cite{Lewenstein1998} is known to equal the overhead, that an entangled state can provide over all separable states in the task of subchannel exclusion~\cite{Ducuara20,UolaBullock20}, and the incompatibility robustness~\cite{Uola15} is known to equal the advantage that incompatible measurements provide over compatible ones in the task of state discrimination~\cite{Carmeli19,Skrzypczyk19,Oszmaniec19,Buscemi20,Uola19}.

So far, the connections between resource quantifiers and quantum tasks have been limited to the case of finite-dimensional resources, with the exception of measurement resources, for which the connection has been recently extended to include also the case of continuous variable systems~\cite{Kuramochi20}. In this manuscript, we provide a method for extending the known finite-dimensional results to the infinite-dimensional setting in the missing cases of quantum channels and quantum states. We perform the extension as a limit procedure, that exploits the known finite-dimensional results. Although the procedure could, in principle, be applied to various quantifiers, we concentrate on the resource measures, which utilise the convex structures in quantum resource theories: the generalized robustness, free robustness, and the convex weight.

Our extension procedure results in resource quantifiers, which enjoy many properties required from a proper resource measure: monotonicity, convexity, and lower semi-continuity. In other words, free operations and convex mixtures do not increase the quantifiers, and the ability of the quantifier to witness resources is stable under small fluctuations. However, the further desirable properties of faithfulness and interpretation as performance in quantum tasks require some more careful consideration.

First, in the finite-dimensional case, faithfulness, as well as the known connections between quantifiers and quantum tasks, require the set of free objects to be closed. This is also expected in the infinite-dimensional case, where, however, the sets might have non-equivalent closures. We show that a natural choice for this closure in our framework arises from the relevant quantum tasks. For sets of free objects that are closed in this sense, the extension procedure results in established faithful infinite-dimensional quantifiers which we show to have an interpretation as performance in some quantum game. It is worth noting that, by our construction, the value of the known quantifiers can be approximated to an arbitrary precision using only finite-dimensional systems.

Second, the extended quantifiers can, in principle, depend on the extension procedure. Whereas this is not the case when the above closure requirement is fulfilled, we introduce another condition on the free set, under which there exists a family of equivalent extension procedures. This requires the possibility of an extension procedure that uses only free operations. Under this condition, the extension procedures result in quantifiers, for which the game interpretation holds. Furthermore, these quantifiers are faithful on the closure of the free set. This is in line with the finite-dimensional setting, in which the game-based quantifiers can be used to separate the free set from objects that are outside of its closure.

In what follows, we give examples of various state and channel resources in the infinite-dimensional setting, that fall under the above requirements on the free set, i.e. they allow for a quantifier with the game interpretation. The examples include entanglement, coherence, asymmetry, the quantum marginal problem, broadcastability (or, more generally, compatibility) of quantum channels, and entanglement breaking channels. Moreover, as the generalized robustness measure is known to be connected to the max-relative entropy, we show that the latter gets an operational interpretation in the infinite-dimensional setting whenever the relative set, i.e. the free set, fulfills the above closure condition.

\textit{Resource quantification in the finite-dimensional setting.---} We concentrate on two main types of resource quantifiers in this manuscript. The first consists of a family of robustness measures $\mathcal R_{F,N}$, where $F$ is the set of free channels, and $N$ is the set of channels that represents the noise. In our applications the noise set will be either all channels (yielding the generalised robustness) or the free set (yielding the free robustness). Robustness measures quantify the amount a given resource channel $\Lambda$ can resist mixing with noise channels $\tilde\Lambda\in N$ before the resource is lost. Formally, we have
\begin{equation} \label{eqn:Robustness}
  \mathcal R_{F,N}(\Lambda) = \min\qty{t\geq 0\,\,\bigg|\,\, \frac{\Lambda + t\tilde\Lambda}{1+t}\in F,\,\tilde\Lambda\in N}.
\end{equation}
Here the free set $F$ is considered to be convex and closed. We note that whereas the choice of the topology in which the set $F$ is closed will be important for our main findings, in the finite-dimensional case all closures are equivalent.

The second type of resource quantifier is the convex weight $W_F$, i.e., the best free approximation of a resource channel $\Lambda$. Formally, we define
\begin{equation} \label{eqn:Weight}
  \mathcal W_F(\Lambda) = \min\qty{\mu\geq 0\,\,\bigg|\,\,\Lambda= \mu\Gamma+(1-\mu)\Lambda_F},
\end{equation}
where the optimization is over all channels $\Gamma$ and all free channels $\Lambda_F\in F$. A possible intuition behind the convex weight is the question of how frequently a free channel $\Lambda_F$ can be used in the preparation procedure of a resource channel $\Lambda$.

Both quantifiers $\mathcal R_{F,N}$ and $\mathcal W_F$ can be cast as conic programs, which allows their evaluation in the dual conic form, see for example~\cite{TakagiRegula19,Uola20} for the robustness and~\cite{Ducuara20,UolaBullock20} for the weight. For the dual formulation, we need the Choi presentation of the channels, i.e. $J_\Lambda:=\frac{1}{d}\sum_{ij}|i\rangle\langle j|\otimes\Lambda(|i\rangle\langle j|)$. In the Choi picture, the duals read
\begin{eqnarray}
\label{eqn:RobSetsChan}
1+\mathcal{R}_F(\Lambda) = \max_{Y}\, && \text{tr}[Y J_{\Lambda}] \\
\text{s.t.: }&& Y \geq 0, \quad \text{tr}[YT]\leq 1\, \forall\ T\in J_{F}, \notag
\end{eqnarray}
and
\begin{eqnarray}
\label{eqn:WeightChanDual}
1-\mathcal{W}_F(\Lambda) = \min_{Y}\, && \text{tr}[Y J_{\Lambda}] \\
\text{s.t.: }&& Y \geq 0, \quad \tr{YT}\geq 1\, \forall\ T\in J_{F}, \notag
\end{eqnarray}
where $Y$ constitutes a witness and $J_{F}$ is the image of the free set $F$ under the Choi isomorphism.

It should be noted, that the evaluation through the dual form holds when the so-called Slater conditions hold. First, this requires the problems to be finite, and, second, for the convex weight, the Slater condition can be verified by choosing $Y=\alpha\openone$ for large enough $\alpha>0$, and for the robustness one requires the existence of a point $\Lambda_F\in F$ and a number $\alpha>1$ such that $\alpha J_{\Lambda_F}-J_\Lambda$ is an interior point of the cone $C_{J_N}$ defined by the noise set $N$. In our noise sets, the existence of some full-rank point in the free set guarantees the Slater condition to hold.

The dual formulation of these resource measures can be used to give them an operational meaning in terms of the performance a channel $\Lambda$ provides in a discrimination~\cite{TakagiRegula19,Uola20} or an exclusion~\cite{UolaBullock20} input-output games. The type of game $\mathcal G$ we are interested in consists of an input set of quantum states $\{\varrho_a\}_a$, i.e., positive unit-trace operators, a quantum measurement $\{M_b\}_b$, that is a positive operator valued measure (or POVM for short), i.e., $M_b\geq0$ for every $b$ and $\sum_b M_b=\openone$, where $\openone$ is the identity operator, and a score assignment $\{\omega_{ab}\}_{ab}$. The payoff $\mathcal P$ of the game $\mathcal G$ for a given channel $\Lambda$ is cast as
\begin{align}\label{eqn:Payoff}
    \mathcal P(\Lambda,\mathcal G)=\sum_{a,b}\omega_{ab}\text{tr}[\Lambda(\varrho_a) M_b].
\end{align}

To relate such games to the resource quantifiers, one can interpret a witness $Y$ as a game. To do so, we write $Y=\sum_{ab}\omega_{ab}\varrho_a^T\otimes M_b$, see Ref.~\cite{Rosset2018}. This way, the object functions of Eq.~(\ref{eqn:RobSetsChan}) and Eq.~(\ref{eqn:WeightChanDual}) are special instances of the payoff function in Eq.~(\ref{eqn:Payoff}). Clearly, a pair of optimal witnesses $Y^r$ for the robustness and $Y^w$ for the weight, with their respective instances of the game $\mathcal G_{Y^r}$ and $\mathcal G_{Y^w}$, lead to the inequalities
\begin{align}
\label{eqn:ChanRobineq}
    \frac{\mathcal P(\Lambda,\mathcal G_{Y^r})}{\max_{\Lambda_F\in F} \mathcal P(\Lambda_F,\mathcal G_{Y^r})}\geq 1+ \mathcal R_{F,N}(\Lambda)
\end{align}
and
\begin{align}
\label{eqn:ChanWeiineq}
    \frac{\mathcal P(\Lambda,\mathcal G_{Y^w})}{\min_{\Lambda_F\in F} \mathcal P(\Lambda_F,\mathcal G_{Y^w})}\leq 1- \mathcal W_F(\Lambda).
\end{align}
One way to make these inequalities tight is to write $\Lambda$ using Eq.~(\ref{eqn:Robustness}) and Eq.~(\ref{eqn:Weight}), and notice that the games are linear in $\Lambda$. This leads to
\begin{align}
\label{eqn:ChanRobAdv}
    \sup_{\mathcal G_N}\frac{\mathcal P(\Lambda,\mathcal G_N)}{\max_{\Lambda_F\in F} \mathcal P(\Lambda_F,\mathcal G_N)}= 1+ \mathcal R_{F,N}(\Lambda)
\end{align}
and
\begin{align}
\label{eqn:ChanWeiAdv}
    \inf_{\mathcal G}\frac{\mathcal P(\Lambda,\mathcal G)}{\min_{\Lambda_F\in F} \mathcal P(\Lambda_F,\mathcal G)}= 1- \mathcal W_F(\Lambda),
\end{align}
where the optimization goes as follows: when the noise set coincides with the set of channels, or one considers the convex weight, the optimization runs over those games $G_N$ that have a non-negative payoff for any channel, and for the case $N=F$, one optimises over those games that have a non-negative payoff for the free set. Also, the games resulting in a zero denominator are excluded.

We note that the above connections between resource quantifiers and quantum games can be also proven for channel tuples. As the extension to this case is straight-forward, i.e. the channels, the respective witnesses, and the games are replaced by tuples, we have spelled the connection in the Appendix~\hyperref[app:tuples]{A}. We further note that quantum states can be seen as quantum channels with a trivial, i.e., one-dimensional, input. Hence, all the above results on quantum channels work also for quantum states. In this case, the corresponding game has a trivial input as well, a fact that is known to render the games into subchannel discrimination~\cite{Takagi19,TakagiRegula19,Uola19,Uola20} or subchannel exclusion tasks~\cite{Ducuara20,UolaBullock20}.

\textit{Resource quantification in the infinite-dimensional setting.---} In this section, we derive the connection between the resource quantifiers and the optimal outperformance in games over free channels for the infinite-dimensional case. We spell out the procedure explicitly using the robustness measures, as the process for the convex weight can be obtained through an even simpler process (as the noise set does not contribute to the convex weight) by interchanging the measures in what follows. The free set of channels $F$ is now a convex subset of the set of all channels (between separable Hilbert spaces $\hil$ and $\mathcal K$). Our method is very general in this framework and only requires that either certain topological constraints are met, or that there is a sound way to restrict the evaluation problem of robustness and convex weight into an approximating sequence of finite-dimensional problems.

There are, in fact, several ways to do the approximation, but we concentrate on a particular method. For this, we assume that there is a channel $\Lambda_F$ in $F$ such that it takes any state into a faithful state, i.e. a state whose eigenvalues are strictly positive. We say that two sequences $(\alpha_n)_n$ and $(\beta_n)_n$ of channels are an approximation if $\alpha_n$ is a channel within a finite-dimensional subspace $\mc H_n$ of $\hil$ and $\beta_n$ is within a finite-dimensional subspace $\mc K_n$ of $\mc K$ and, for any states $\varrho$ on $\hil$ and $\sigma$ on $\mc K$, $\|\varrho-\alpha_n(\varrho)\|_{\rm tr}\to0$ and $\|\sigma-\beta_n(\sigma)\|_{\rm tr}\to0$ as $n\to\infty$, whenever $\varrho$ is supported by $\mc H_n$ then $\alpha_n(\varrho)=\varrho$ and similarly for a state $\sigma$ supported by $\mc K_n$, and, if $\varrho$ is faithful, then $\beta_n(\varrho)$ is of full rank within $\mc K_n$. Moreover, we require that, when $m\leq n$, $\alpha_n\circ\alpha_m=\alpha_m$ and $\beta_m\circ\beta_n=\beta_m$. As an example, whenever $(\mc H_n)_n$ is an increasing sequence of finite-dimensional subspaces of $\hil$ the closure of the union of which is $\hil$, the channels $\alpha_n(\varrho)=P_n\varrho P_n+{\rm tr}[\varrho P_n^\perp]\varrho_0$ where $P_n$ is the orthogonal projection onto $\mc H_n$ and $\varrho_0$ is some fixed state within $\mc H_1$, and a similar choice for $\beta_n$ provide an example of an approximation.

Suppose that $\mathbb A=\{\alpha_n,\beta_n\}_n$ is an approximation procedure and denote $F_n:=\{\beta_n\circ\Lambda_F\circ\alpha_n\,|\,\Lambda_F\in F\}$ and $N_n:=\{\beta_n\circ\Lambda_N\circ\alpha_n\,|\,\Lambda_N\in N\}$. We denote the standard closures of these finite-dimensional sets by $\overline{F}_n$ and $\overline{N}_n$. We easily see (see Appendix~\hyperref[app:Properties]{B}) that the sequence $\big(\mc R^{\mathbb A}_{\overline{F}_n,\overline{N}_n}(\beta_n\circ\Lambda\circ\alpha_n)\big)_{n}$ is increasing. One can now define the approximate robustness $\mathcal R_{F,N}^{\mathbb A}$ as the supremum over $n\in\N$ (which, according to the above is also the limit as $n\to\infty$) of the individual robustnesses $\mc R_{\overline{F}_n,\overline{N}_n}(\beta_n \circ\Lambda\circ\alpha_n)$. As a pointwise supremum of a family of convex lower semi-continuous functions, $\mc R_{F,N}^{\mathbb A}$ is also convex and lower semi-continuous. Also, $\mc R_{F,N}^{\mathbb A}$ is non-increasing under operations that do not map elements of $F$ outside of $F$ (and same for $N$), see Appendix~\hyperref[app:Properties]{B}. Another property usually required from a resource measure is faithfulness, i.e., $\mc R_{F,N}^{\mathbb A}(\Lambda)=0$ if and only if $\Lambda\in F$. As this property depends on the topological properties of the set $F$, we will comment on this later. We note, that this kind of dependency is not specific to the approximate robustness only, but it also affects the original one $\mc R_{F,N}$. 

The results from the previous section on finite-dimensional resources apply to any convex and compact free set $F$, when $N$ is either $F$ or the whole set of channels. As the set $N$ maps surjectively to the corresponding set $N_n$ of channels (free or whole) from $\mc H_n$ to $\mc K_n$, the robustness measure $\mc R_{\overline{F}_n,\overline{N}_n}$ has a game interpretation for each $n$. More specifically, Eq.~(\ref{eqn:ChanRobAdv}) leads to
\begin{align}
\label{eqn:ChanRobsupadv}
    \sup_n\sup_{\mathcal G_n}\frac{\mathcal P(\beta_n\circ\Lambda\circ\alpha_n,\mathcal G_n)}{\sup_{\Lambda_F\in F_n} \mathcal P(\Lambda_F,\mathcal G_n)}= 1+ \mathcal R^{\mathbb A}_{F,N}(\Lambda)
\end{align}
where the second supremum runs over all games $\mc G_n$ on the set of channels between $\mc H_n$ and $\mc K_n$ such that $\mc P(\Lambda_n,\mc G_n)\geq0$ for all $\Lambda_n\in N_n$.

We note, that in principle the quantity $\mathcal R^{\mathbb A}_{F,N}$ can depend on the approximation procedure. However, it is always a lower bound on the actual robustness, as defined in Eq.~(\ref{eqn:Robustness}), and there is a simple sufficient condition for these robustnesses to agree. Namely, if the sets $F$ and $N$ are suitably closed, then the approximation procedure reaches $R_{F,N}$.

A natural choice for the topology is the one related to the games. More precisely, one can ask for any given input, i.e., a trace-class operator $\varrho_a$, and any given output , i.e., a bounded operator $M_b$, whether a sequence (or a net) of channels $(\Lambda_n)_n$ converges to a channel $\Lambda$ in the sense that $\text{tr}[\Lambda_n(\varrho_a)M_b]\to{\rm tr}[\Lambda(\varrho_a)M_b]$ as $n\to\infty$. We denote the topology associated with this type of convergence by $\tau$. In fact, we see that when $F$ is closed w.r.t $\tau$ and $N$ is the whole set of channels or $N$ is $\tau$-closed, these robustnesses coincide, i.e.,\ for any approximation $\mathbb A$, $\mc R^{\mathbb A}_{F,N}=\mc R_{F,N}$. Moreover, in this case it is easy to see, that the approximate robustness (as well as the original one) are faithful. We summarize these ideas in the following Theorem, the detailed proof of which is presented in the Appendix~\hyperref[app:Proof1]{C}.

\begin{theorem}\label{Theorem1}
Let the free set be $\tau$-closed. Whenever $N$ is $\tau$-closed or the whole set of channels, we have $\mathcal R^{\mathbb A}_{F,N}(\Lambda)=\mathcal R_{F,N}(\Lambda)$, and the analogically defined approximate weight $\mc W_F^{\mathbb A}(\Lambda)=\mc W_F(\Lambda)$ for any approximation procedure, and all quantifiers are faithful. If, furthermore, $N=F$ or $N$ is the whole set we have
\begin{align}
    \label{eqn:ChanRobinfadv}
    \sup_{\mathcal G}\frac{\mathcal P(\Lambda,\mathcal G)}{\sup_{\Lambda_F\in F} \mathcal P(\Lambda_F,\mathcal G)}= 1+ \mathcal R_{F,N}(\Lambda),
\end{align}
where the outer supremum runs over those games that have a positive payoff in the set $N$, whenever the right-hand-side is finite. Moreover, one has
\begin{align}
\label{eqn:ChanWeiinfAdv}
    \inf_{\mathcal G}\frac{\mathcal P(\Lambda,\mathcal G)}{\inf_{\Lambda_F\in F} \mathcal P(\Lambda_F,\mathcal G)}= 1- \mathcal W_F(\Lambda),
\end{align}
where the infimum runs over those games that have a positive payoff for any channel, and we omit the games that result in a zero denominator.
\end{theorem}

As a straight-forward application of our Theorem \ref{Theorem1}, one can take channels with a trivial, i.e. one-dimensional, input. Such channels correspond to quantum states. In this case, the topology $\tau$ reduces to the well-known $\sigma$-weak topology generated by bounded operators. Hence, any infinite-dimensional state resource that has a $\sigma$-weakly closed set of free states, can be related to performance in quantum games. In finite-dimensional setting, state resources are known to relate to subchannel discrimination (robustness) and exclusion (weight) tasks. As every approximation step is finite-dimensional, the related games in Theorem \ref{Theorem1} are also subchannel discrimination and exclusion tasks.

The conditions of Theorem~\ref{Theorem1} should be checked for each free set separately, which is not our emphasis here. Examples include incoherent and symmetric states in the case of a compact group, see Appendix~\hyperref[app:Proof1]{C}. However, we give another sufficient condition under which the approximated quantifiers satisfy the counterparts of Eq.~(\ref{eqn:ChanRobinfadv}) and Eq.~(\ref{eqn:ChanWeiinfAdv}). In such circumstances, the approximate robustness has again a clear operational interpretation given by outperformance in general (i.e. possibly infinite-dimensional) games, c.f. Eq.~(\ref{eqn:ChanRobsupadv}) where the advantage is restricted to finite games. This is summarized in the following Observation, the proof of which is given in Appendix~\hyperref[app:Proof2]{D}.

\begin{observation}\label{Obs1}
Let $N$ be the free set $F$ or the whole set and the approximation procedure $\mathbb A=\{\alpha_n,\beta_n\}_n$ be such that $F_n\subseteq F$ and $N_n\subseteq N$ for all $n$. Now Eq.~(\ref{eqn:ChanRobinfadv}) and Eq.~(\ref{eqn:ChanWeiinfAdv}) hold when one replaces the robustness $\mc R_{F,N}$ by $\mc R_{F,N}^{\mathbb A}$ and $\mc W_F$ by $\mc W_F^{\mathbb A}$. In this case we do not require $\tau$-closedness of the free set $F$. In particular, the approximate quantifiers are independent of the chosen approximation procedure $\mathbb A$, and they lower bound the robustness $\mc R_{F,N}$ and the weight $\mc W_N$ respectively.
\end{observation}

In the setting of the above Observation, the approximate quantifiers correspond to the extensions of the original quantifiers with respect to the $\tau$-closure of the free set, see Appendix~\hyperref[app:Proof2]{D}. It is natural to consider such extensions, as faithfulness is difficult to establish for resource measures when the free set is not closed. It follows that the approximate quantifiers are faithful with respect to the closure. Especially, they are faithful with respect to $F$ itself, whenever the closure does not introduce any physical states or channels in addition to the original ones in $F$. Here states that do respect the Born rule, and channels that have a representation in the Schrödinger picture are called physical.

Observation~\ref{Obs1} opens up various resources for a game interpretation without having to dive into their topological properties. For example, the free sets associated with the quantum marginal problem, entanglement, incompatibility of channels, and entanglement breaking channels are easily seen to allow an approximation procedure of the desired type, see Appendix~\hyperref[app:Proof2]{D}. Of course, the quantum marginal problem (resp. incompatibility of quantum channels) is defined on tuples of states (resp. channels), which {\it per se} are not within the scope of our results. However, our results can be easily extended to tuples by slight modifications of the proofs and the games, see Appendix~\hyperref[app:Proof2]{D}.

\textit{Connection to relative entropy.---} In the finite-dimensional setting, the generalised robustness measure is known to relate to the max-relative entropy. This connection can be easily extended to the realm of our work. The max-relative entropy is defined by $D_{max}(\varrho\|\sigma) = \log \inf\qty{\lambda \vert \varrho \leq \lambda \sigma}$, for positive trace-class operators $\varrho,\sigma\geq 0$, with $\text{tr}[\varrho]=1$, and $\text{supp}(\varrho)\subseteq\text{supp}(\sigma)$. Let $F$ be the convex and (trace-norm-closed) set of free states, the \emph{max entropy} of $\varrho$ with respect to the set $F$ is defined as $E_{max}(\varrho) = \inf_{\sigma\in F} D_{max}(\varrho\|\sigma)$. We obtain the following observation:
\begin{observation}\label{Obs2}
Let $\varrho$ and $\sigma$ be quantum states, then $E_{max}(\varrho)= \log(1+R_F(\rho))$.
\end{observation}
\begin{proof}
The proof is similar as in Ref.~\cite{Datta2009b}. First, observe that the robustness can be written as
\begin{align}
R_F(\varrho)&=\inf \{t\geq0|\varrho+t\tau=(1+t)\sigma, \sigma\in F\}\nonumber\\
&=\inf \{t\geq0|\varrho\leq (1+t)\sigma, \sigma\in F\}.
\end{align}
Then one finds that
\begin{align}
E_{max}(\varrho)&=\inf_{\sigma\in F} \log \inf\qty{\lambda\geq 1 \vert \varrho \leq \lambda \sigma}\nonumber\\
&= \log \inf_{\sigma\in F} \inf\qty{\lambda\geq 1 \vert \varrho \leq \lambda \sigma}\nonumber\\
&=\log  \inf\qty{\lambda\geq 1 \vert \varrho \leq \lambda \sigma, \sigma\in F}\nonumber\\
&=\log(1+R_F(\varrho)),
\end{align}
where the fist equality is due to the concavity of the logarithm and the last equality is obtained by $\lambda\mapsto 1+t$.
\end{proof}

By recalling Theorem~\ref{Theorem1}, we see that the max-relative entropy has an interpretation through quantum games, whenever the free set of states is $\tau$-closed.

\textit{Conclusions.---} We have presented a method for extending finite-dimensional quantum resource quantifiers into the infinite-dimensional regime. We have demonstrated our technique using the generalized robustness, the free robustness, and the convex weight. In the case of quantum states (resp. quantum channels) these quantifiers relate to performance in discrimination tasks (resp. quantum games) in the finite-dimensional setting. We have identified sufficient conditions on the free sets, under which the approximations give the performance interpretation for the established infinite-dimensional resource quantifiers, and for the natural extensions thereof. We have presented various examples of quantum state and channel resources that fall under these sufficient conditions. Moreover, we show that the connection between max-relative entropy and operational tasks established in~\cite{Bu2017} carries over to the infinite-dimensional setting.

For future research, it will be interesting to find further conditions under which the free sets allow faithful quantifiers with the game interpretation. Moreover, an open question is to identify free sets where the the approximate robustness fails to coincide with the original one, i.e. where there is a finite gap between the two. Furthermore, our work paves way to the resource theory of more general dynamical objects, such as quantum instruments, and more specialised quantifiers therein, such as tolerance against specific noise sets, in continuous variable systems.

\textit{Acknowledgements.---}
E.H. acknowledges financial support from the National Natural Science Foundation of China (Grant No. 11875110). T.K. acknowledges support from the ERC (Consolidator Grant No.~683107/TempoQ) and  
the Deutsche Forschungsgemeinschaft (DFG, German Research Foundation - 447948357). R.U. acknowledges financial support from the Swiss National Science Foundation (Starting grant DIAQ and NCCR QSIT).

\textit{Note added.---}
Recently, two articles by Lami, Regula, Takagi, and Ferrari~\cite{regula20,lami20} presented related results on quantum state resources. In their work, they presented a sufficient condition for the closing of the gap between their counterpart of our approximate robustness and the original robustness measure. In our setting that condition is as follows: We define the topology $\tau_0$ on the set of non-normalized channels, i.e.,\ the space $\mc V$ of completely bounded linear maps as the coarsest topology w.r.t.\ which the maps $\mc V\ni\Lambda\mapsto\tr{\Lambda(\varrho)K}\in\C$ are continuous for all input states $\varrho$ and compact operators on the output space $\mc K$. The condition is that the cone corresponding to $F$ (i.e.,\ the cone whose intersection with the set of channels coincides with $F$) is $\tau_0$-closed. (Use of cones here is motivated by noting that $F$ itself is typically not $\tau_0$-closed as the trace-preservation condition for channels is problematic in this topology.) As $\tau_0$ is coarser than our topology $\tau$, this cone is also $\tau$-closed and so is $F$ as the intersection of this cone with the $\tau$-closed set of completely positive unital linear maps in the Heisenberg picture (possibly without a Schr\"odinger description). Thus this condition implies our condition in Theorem \ref{Theorem1}. As an example, the authors proved that the cone of separable states is closed in the reduction of $\tau_0$-topology to the case of quantum states, i.e. in the $\sigma$-weak topology generated by the compact operators. This shows, that the set of separable states is $\tau$-closed and, consequently, falls into the realm of our Theorem~\ref{Theorem1}.

\bibliography{bibliography.bib}

\begin{thebibliography}{35}%
\makeatletter
\providecommand \@ifxundefined [1]{%
 \@ifx{#1\undefined}
}%
\providecommand \@ifnum [1]{%
 \ifnum #1\expandafter \@firstoftwo
 \else \expandafter \@secondoftwo
 \fi
}%
\providecommand \@ifx [1]{%
 \ifx #1\expandafter \@firstoftwo
 \else \expandafter \@secondoftwo
 \fi
}%
\providecommand \natexlab [1]{#1}%
\providecommand \enquote  [1]{``#1''}%
\providecommand \bibnamefont  [1]{#1}%
\providecommand \bibfnamefont [1]{#1}%
\providecommand \citenamefont [1]{#1}%
\providecommand \href@noop [0]{\@secondoftwo}%
\providecommand \href [0]{\begingroup \@sanitize@url \@href}%
\providecommand \@href[1]{\@@startlink{#1}\@@href}%
\providecommand \@@href[1]{\endgroup#1\@@endlink}%
\providecommand \@sanitize@url [0]{\catcode `\\12\catcode `\$12\catcode
  `\&12\catcode `\#12\catcode `\^12\catcode `\_12\catcode `\%12\relax}%
\providecommand \@@startlink[1]{}%
\providecommand \@@endlink[0]{}%
\providecommand \url  [0]{\begingroup\@sanitize@url \@url }%
\providecommand \@url [1]{\endgroup\@href {#1}{\urlprefix }}%
\providecommand \urlprefix  [0]{URL }%
\providecommand \Eprint [0]{\href }%
\providecommand \doibase [0]{http://dx.doi.org/}%
\providecommand \selectlanguage [0]{\@gobble}%
\providecommand \bibinfo  [0]{\@secondoftwo}%
\providecommand \bibfield  [0]{\@secondoftwo}%
\providecommand \translation [1]{[#1]}%
\providecommand \BibitemOpen [0]{}%
\providecommand \bibitemStop [0]{}%
\providecommand \bibitemNoStop [0]{.\EOS\space}%
\providecommand \EOS [0]{\spacefactor3000\relax}%
\providecommand \BibitemShut  [1]{\csname bibitem#1\endcsname}%
\let\auto@bib@innerbib\@empty
\bibitem [{\citenamefont {Chitambar}\ and\ \citenamefont
  {Gour}(2019)}]{ReviewQRT}%
  \BibitemOpen
  \bibfield  {author} {\bibinfo {author} {\bibfnamefont {E.}~\bibnamefont
  {Chitambar}}\ and\ \bibinfo {author} {\bibfnamefont {G.}~\bibnamefont
  {Gour}},\ }\href {\doibase 10.1103/RevModPhys.91.025001} {\bibfield
  {journal} {\bibinfo  {journal} {Rev. Mod. Phys.}\ }\textbf {\bibinfo {volume}
  {91}},\ \bibinfo {pages} {025001} (\bibinfo {year} {2019})}\BibitemShut
  {NoStop}%
\bibitem [{\citenamefont {Datta}(2009{\natexlab{a}})}]{Datta2009a}%
  \BibitemOpen
  \bibfield  {author} {\bibinfo {author} {\bibfnamefont {N.}~\bibnamefont
  {Datta}},\ }\href {\doibase 10.1109/tit.2009.2018325} {\bibfield  {journal}
  {\bibinfo  {journal} {IEEE Trans. Inf. Theory}\ }\textbf {\bibinfo {volume}
  {55}},\ \bibinfo {pages} {2816–2826} (\bibinfo {year}
  {2009}{\natexlab{a}})}\BibitemShut {NoStop}%
\bibitem [{\citenamefont {Datta}(2009{\natexlab{b}})}]{Datta2009b}%
  \BibitemOpen
  \bibfield  {author} {\bibinfo {author} {\bibfnamefont {N.}~\bibnamefont
  {Datta}},\ }\href {\doibase 10.1142/S0219749909005298} {\bibfield  {journal}
  {\bibinfo  {journal} {Int. J. Quantum Inf.}\ }\textbf {\bibinfo {volume}
  {07}},\ \bibinfo {pages} {475} (\bibinfo {year}
  {2009}{\natexlab{b}})}\BibitemShut {NoStop}%
\bibitem [{\citenamefont {Napoli}\ \emph {et~al.}(2016)\citenamefont {Napoli},
  \citenamefont {Bromley}, \citenamefont {Cianciaruso}, \citenamefont {Piani},
  \citenamefont {Johnston},\ and\ \citenamefont {Adesso}}]{NBC+16}%
  \BibitemOpen
  \bibfield  {author} {\bibinfo {author} {\bibfnamefont {C.}~\bibnamefont
  {Napoli}}, \bibinfo {author} {\bibfnamefont {T.~R.}\ \bibnamefont {Bromley}},
  \bibinfo {author} {\bibfnamefont {M.}~\bibnamefont {Cianciaruso}}, \bibinfo
  {author} {\bibfnamefont {M.}~\bibnamefont {Piani}}, \bibinfo {author}
  {\bibfnamefont {N.}~\bibnamefont {Johnston}}, \ and\ \bibinfo {author}
  {\bibfnamefont {G.}~\bibnamefont {Adesso}},\ }\href {\doibase
  10.1103/PhysRevLett.116.150502} {\bibfield  {journal} {\bibinfo  {journal}
  {Phys. Rev. Lett.}\ }\textbf {\bibinfo {volume} {116}},\ \bibinfo {pages}
  {150502} (\bibinfo {year} {2016})}\BibitemShut {NoStop}%
\bibitem [{\citenamefont {Skrzypczyk}\ and\ \citenamefont
  {Linden}(2019)}]{SL19}%
  \BibitemOpen
  \bibfield  {author} {\bibinfo {author} {\bibfnamefont {P.}~\bibnamefont
  {Skrzypczyk}}\ and\ \bibinfo {author} {\bibfnamefont {N.}~\bibnamefont
  {Linden}},\ }\href {\doibase 10.1103/physrevlett.122.140403} {\bibfield
  {journal} {\bibinfo  {journal} {Physical Review Letters}\ }\textbf {\bibinfo
  {volume} {122}} (\bibinfo {year} {2019}),\
  10.1103/physrevlett.122.140403}\BibitemShut {NoStop}%
\bibitem [{\citenamefont {Piani}\ and\ \citenamefont
  {Watrous}(2015)}]{Piani15}%
  \BibitemOpen
  \bibfield  {author} {\bibinfo {author} {\bibfnamefont {M.}~\bibnamefont
  {Piani}}\ and\ \bibinfo {author} {\bibfnamefont {J.}~\bibnamefont
  {Watrous}},\ }\href {\doibase 10.1103/PhysRevLett.114.060404} {\bibfield
  {journal} {\bibinfo  {journal} {Phys. Rev. Lett.}\ }\textbf {\bibinfo
  {volume} {114}},\ \bibinfo {pages} {060404} (\bibinfo {year}
  {2015})}\BibitemShut {NoStop}%
\bibitem [{\citenamefont {Takagi}\ \emph {et~al.}(2019)\citenamefont {Takagi},
  \citenamefont {Regula}, \citenamefont {Bu}, \citenamefont {Liu},\ and\
  \citenamefont {Adesso}}]{Takagi19}%
  \BibitemOpen
  \bibfield  {author} {\bibinfo {author} {\bibfnamefont {R.}~\bibnamefont
  {Takagi}}, \bibinfo {author} {\bibfnamefont {B.}~\bibnamefont {Regula}},
  \bibinfo {author} {\bibfnamefont {K.}~\bibnamefont {Bu}}, \bibinfo {author}
  {\bibfnamefont {Z.-W.}\ \bibnamefont {Liu}}, \ and\ \bibinfo {author}
  {\bibfnamefont {G.}~\bibnamefont {Adesso}},\ }\href {\doibase
  10.1103/physrevlett.122.140402} {\bibfield  {journal} {\bibinfo  {journal}
  {Phys. Rev. Lett.}\ }\textbf {\bibinfo {volume} {122}},\ \bibinfo {pages}
  {140402} (\bibinfo {year} {2019})}\BibitemShut {NoStop}%
\bibitem [{\citenamefont {Bae}\ \emph {et~al.}(2019)\citenamefont {Bae},
  \citenamefont {Chru\ifmmode \acute{s}\else
  \'{s}\fi{}ci\ifmmode~\acute{n}\else \'{n}\fi{}ski},\ and\ \citenamefont
  {Piani}}]{BCP19}%
  \BibitemOpen
  \bibfield  {author} {\bibinfo {author} {\bibfnamefont {J.}~\bibnamefont
  {Bae}}, \bibinfo {author} {\bibfnamefont {D.}~\bibnamefont {Chru\ifmmode
  \acute{s}\else \'{s}\fi{}ci\ifmmode~\acute{n}\else \'{n}\fi{}ski}}, \ and\
  \bibinfo {author} {\bibfnamefont {M.}~\bibnamefont {Piani}},\ }\href
  {\doibase 10.1103/PhysRevLett.122.140404} {\bibfield  {journal} {\bibinfo
  {journal} {Phys. Rev. Lett.}\ }\textbf {\bibinfo {volume} {122}},\ \bibinfo
  {pages} {140404} (\bibinfo {year} {2019})}\BibitemShut {NoStop}%
\bibitem [{\citenamefont {Takagi}\ and\ \citenamefont
  {Regula}(2019)}]{TakagiRegula19}%
  \BibitemOpen
  \bibfield  {author} {\bibinfo {author} {\bibfnamefont {R.}~\bibnamefont
  {Takagi}}\ and\ \bibinfo {author} {\bibfnamefont {B.}~\bibnamefont
  {Regula}},\ }\href {\doibase 10.1103/PhysRevX.9.031053} {\bibfield  {journal}
  {\bibinfo  {journal} {Phys. Rev. X}\ }\textbf {\bibinfo {volume} {9}},\
  \bibinfo {pages} {031053} (\bibinfo {year} {2019})}\BibitemShut {NoStop}%
\bibitem [{\citenamefont {Lipka-Bartosik}\ and\ \citenamefont
  {Skrzypczyk}(2019)}]{LPS19}%
  \BibitemOpen
  \bibfield  {author} {\bibinfo {author} {\bibfnamefont {P.}~\bibnamefont
  {Lipka-Bartosik}}\ and\ \bibinfo {author} {\bibfnamefont {P.}~\bibnamefont
  {Skrzypczyk}},\ }\href {\doibase 10.1103/PhysRevResearch.2.023029} {\enquote
  {\bibinfo {title} {{T}he operational advantages provided by non-classical
  teleportation},}\ }\bibinfo {howpublished} {Phys. Rev. Research 2, 023029
  (2020)} (\bibinfo {year} {2019}),\ \bibinfo {note} {arXiv:1908.05107v2},\
  \Eprint {http://arxiv.org/abs/1908.05107} {1908.05107} \BibitemShut {NoStop}%
\bibitem [{\citenamefont {Ducuara}\ \emph {et~al.}(2020)\citenamefont
  {Ducuara}, \citenamefont {Lipka-Bartosik},\ and\ \citenamefont
  {Skrzypczyk}}]{DLPS20}%
  \BibitemOpen
  \bibfield  {author} {\bibinfo {author} {\bibfnamefont {A.~F.}\ \bibnamefont
  {Ducuara}}, \bibinfo {author} {\bibfnamefont {P.}~\bibnamefont
  {Lipka-Bartosik}}, \ and\ \bibinfo {author} {\bibfnamefont {P.}~\bibnamefont
  {Skrzypczyk}},\ }\href {\doibase 10.1103/physrevresearch.2.033374} {\bibfield
   {journal} {\bibinfo  {journal} {Physical Review Research}\ }\textbf
  {\bibinfo {volume} {2}} (\bibinfo {year} {2020}),\
  10.1103/physrevresearch.2.033374}\BibitemShut {NoStop}%
\bibitem [{\citenamefont {Uola}\ \emph
  {et~al.}(2020{\natexlab{a}})\citenamefont {Uola}, \citenamefont {Kraft},\
  and\ \citenamefont {Abbott}}]{Uola20}%
  \BibitemOpen
  \bibfield  {author} {\bibinfo {author} {\bibfnamefont {R.}~\bibnamefont
  {Uola}}, \bibinfo {author} {\bibfnamefont {T.}~\bibnamefont {Kraft}}, \ and\
  \bibinfo {author} {\bibfnamefont {A.~A.}\ \bibnamefont {Abbott}},\ }\href
  {\doibase 10.1103/PhysRevA.101.052306} {\bibfield  {journal} {\bibinfo
  {journal} {Phys. Rev. A}\ }\textbf {\bibinfo {volume} {101}},\ \bibinfo
  {pages} {052306} (\bibinfo {year} {2020}{\natexlab{a}})}\BibitemShut
  {NoStop}%
\bibitem [{\citenamefont {Lewenstein}\ and\ \citenamefont
  {Sanpera}(1998)}]{Lewenstein1998}%
  \BibitemOpen
  \bibfield  {author} {\bibinfo {author} {\bibfnamefont {M.}~\bibnamefont
  {Lewenstein}}\ and\ \bibinfo {author} {\bibfnamefont {A.}~\bibnamefont
  {Sanpera}},\ }\href {\doibase 10.1103/PhysRevLett.80.2261} {\bibfield
  {journal} {\bibinfo  {journal} {Phys. Rev. Lett.}\ }\textbf {\bibinfo
  {volume} {80}},\ \bibinfo {pages} {2261} (\bibinfo {year}
  {1998})}\BibitemShut {NoStop}%
\bibitem [{\citenamefont {Ducuara}\ and\ \citenamefont
  {Skrzypczyk}(2020)}]{Ducuara20}%
  \BibitemOpen
  \bibfield  {author} {\bibinfo {author} {\bibfnamefont {A.~F.}\ \bibnamefont
  {Ducuara}}\ and\ \bibinfo {author} {\bibfnamefont {P.}~\bibnamefont
  {Skrzypczyk}},\ }\href {\doibase 10.1103/PhysRevLett.125.110401} {\bibfield
  {journal} {\bibinfo  {journal} {Phys. Rev. Lett.}\ }\textbf {\bibinfo
  {volume} {125}},\ \bibinfo {pages} {110401} (\bibinfo {year}
  {2020})}\BibitemShut {NoStop}%
\bibitem [{\citenamefont {Uola}\ \emph
  {et~al.}(2020{\natexlab{b}})\citenamefont {Uola}, \citenamefont {Bullock},
  \citenamefont {Kraft}, \citenamefont {Pellonp\"a\"a},\ and\ \citenamefont
  {Brunner}}]{UolaBullock20}%
  \BibitemOpen
  \bibfield  {author} {\bibinfo {author} {\bibfnamefont {R.}~\bibnamefont
  {Uola}}, \bibinfo {author} {\bibfnamefont {T.}~\bibnamefont {Bullock}},
  \bibinfo {author} {\bibfnamefont {T.}~\bibnamefont {Kraft}}, \bibinfo
  {author} {\bibfnamefont {J.-P.}\ \bibnamefont {Pellonp\"a\"a}}, \ and\
  \bibinfo {author} {\bibfnamefont {N.}~\bibnamefont {Brunner}},\ }\href
  {\doibase 10.1103/PhysRevLett.125.110402} {\bibfield  {journal} {\bibinfo
  {journal} {Phys. Rev. Lett.}\ }\textbf {\bibinfo {volume} {125}},\ \bibinfo
  {pages} {110402} (\bibinfo {year} {2020}{\natexlab{b}})}\BibitemShut
  {NoStop}%
\bibitem [{\citenamefont {Uola}\ \emph {et~al.}(2015)\citenamefont {Uola},
  \citenamefont {Budroni}, \citenamefont {Gühne},\ and\ \citenamefont
  {Pellonpää}}]{Uola15}%
  \BibitemOpen
  \bibfield  {author} {\bibinfo {author} {\bibfnamefont {R.}~\bibnamefont
  {Uola}}, \bibinfo {author} {\bibfnamefont {C.}~\bibnamefont {Budroni}},
  \bibinfo {author} {\bibfnamefont {O.}~\bibnamefont {Gühne}}, \ and\ \bibinfo
  {author} {\bibfnamefont {J.-P.}\ \bibnamefont {Pellonpää}},\ }\href
  {\doibase 10.1103/physrevlett.115.230402} {\bibfield  {journal} {\bibinfo
  {journal} {Phys. Rev. Lett.}\ }\textbf {\bibinfo {volume} {115}} (\bibinfo
  {year} {2015}),\ 10.1103/physrevlett.115.230402}\BibitemShut {NoStop}%
\bibitem [{\citenamefont {Carmeli}\ \emph {et~al.}(2019)\citenamefont
  {Carmeli}, \citenamefont {Heinosaari},\ and\ \citenamefont
  {Toigo}}]{Carmeli19}%
  \BibitemOpen
  \bibfield  {author} {\bibinfo {author} {\bibfnamefont {C.}~\bibnamefont
  {Carmeli}}, \bibinfo {author} {\bibfnamefont {T.}~\bibnamefont {Heinosaari}},
  \ and\ \bibinfo {author} {\bibfnamefont {A.}~\bibnamefont {Toigo}},\ }\href
  {\doibase 10.1103/physrevlett.122.130402} {\bibfield  {journal} {\bibinfo
  {journal} {Phys. Rev. Lett.}\ }\textbf {\bibinfo {volume} {122}} (\bibinfo
  {year} {2019}),\ 10.1103/physrevlett.122.130402}\BibitemShut {NoStop}%
\bibitem [{\citenamefont {Skrzypczyk}\ \emph {et~al.}(2019)\citenamefont
  {Skrzypczyk}, \citenamefont {Šupić},\ and\ \citenamefont
  {Cavalcanti}}]{Skrzypczyk19}%
  \BibitemOpen
  \bibfield  {author} {\bibinfo {author} {\bibfnamefont {P.}~\bibnamefont
  {Skrzypczyk}}, \bibinfo {author} {\bibfnamefont {I.}~\bibnamefont {Šupić}},
  \ and\ \bibinfo {author} {\bibfnamefont {D.}~\bibnamefont {Cavalcanti}},\
  }\href {\doibase 10.1103/physrevlett.122.130403} {\bibfield  {journal}
  {\bibinfo  {journal} {Phys. Rev. Lett.}\ }\textbf {\bibinfo {volume} {122}}
  (\bibinfo {year} {2019}),\ 10.1103/physrevlett.122.130403}\BibitemShut
  {NoStop}%
\bibitem [{\citenamefont {Oszmaniec}\ and\ \citenamefont
  {Biswas}(2019)}]{Oszmaniec19}%
  \BibitemOpen
  \bibfield  {author} {\bibinfo {author} {\bibfnamefont {M.}~\bibnamefont
  {Oszmaniec}}\ and\ \bibinfo {author} {\bibfnamefont {T.}~\bibnamefont
  {Biswas}},\ }\href {\doibase 10.22331/q-2019-04-26-133} {\bibfield  {journal}
  {\bibinfo  {journal} {Quantum}\ }\textbf {\bibinfo {volume} {3}},\ \bibinfo
  {pages} {133} (\bibinfo {year} {2019})}\BibitemShut {NoStop}%
\bibitem [{\citenamefont {Buscemi}\ \emph {et~al.}(2020)\citenamefont
  {Buscemi}, \citenamefont {Chitambar},\ and\ \citenamefont
  {Zhou}}]{Buscemi20}%
  \BibitemOpen
  \bibfield  {author} {\bibinfo {author} {\bibfnamefont {F.}~\bibnamefont
  {Buscemi}}, \bibinfo {author} {\bibfnamefont {E.}~\bibnamefont {Chitambar}},
  \ and\ \bibinfo {author} {\bibfnamefont {W.}~\bibnamefont {Zhou}},\ }\href
  {\doibase 10.1103/physrevlett.124.120401} {\bibfield  {journal} {\bibinfo
  {journal} {Phys. Rev. Lett.}\ }\textbf {\bibinfo {volume} {124}} (\bibinfo
  {year} {2020}),\ 10.1103/physrevlett.124.120401}\BibitemShut {NoStop}%
\bibitem [{\citenamefont {Uola}\ \emph {et~al.}(2019)\citenamefont {Uola},
  \citenamefont {Kraft}, \citenamefont {Shang}, \citenamefont {Yu},\ and\
  \citenamefont {G{\"u}hne}}]{Uola19}%
  \BibitemOpen
  \bibfield  {author} {\bibinfo {author} {\bibfnamefont {R.}~\bibnamefont
  {Uola}}, \bibinfo {author} {\bibfnamefont {T.}~\bibnamefont {Kraft}},
  \bibinfo {author} {\bibfnamefont {J.}~\bibnamefont {Shang}}, \bibinfo
  {author} {\bibfnamefont {X.-D.}\ \bibnamefont {Yu}}, \ and\ \bibinfo {author}
  {\bibfnamefont {O.}~\bibnamefont {G{\"u}hne}},\ }\href {\doibase
  10.1103/PhysRevLett.122.130404} {\bibfield  {journal} {\bibinfo  {journal}
  {Phys. Rev. Lett.}\ }\textbf {\bibinfo {volume} {122}},\ \bibinfo {pages}
  {130404} (\bibinfo {year} {2019})}\BibitemShut {NoStop}%
\bibitem [{\citenamefont {Kuramochi}(2020)}]{Kuramochi20}%
  \BibitemOpen
  \bibfield  {author} {\bibinfo {author} {\bibfnamefont {Y.}~\bibnamefont
  {Kuramochi}},\ }\href@noop {} {\  (\bibinfo {year} {2020})},\ \Eprint
  {http://arxiv.org/abs/2002.03504} {arXiv:2002.03504 [math.FA]} \BibitemShut
  {NoStop}%
\bibitem [{\citenamefont {Rosset}\ \emph {et~al.}(2018)\citenamefont {Rosset},
  \citenamefont {Buscemi},\ and\ \citenamefont {Liang}}]{Rosset2018}%
  \BibitemOpen
  \bibfield  {author} {\bibinfo {author} {\bibfnamefont {D.}~\bibnamefont
  {Rosset}}, \bibinfo {author} {\bibfnamefont {F.}~\bibnamefont {Buscemi}}, \
  and\ \bibinfo {author} {\bibfnamefont {Y.-C.}\ \bibnamefont {Liang}},\ }\href
  {\doibase 10.1103/PhysRevX.8.021033} {\bibfield  {journal} {\bibinfo
  {journal} {Phys. Rev. X}\ }\textbf {\bibinfo {volume} {8}},\ \bibinfo {pages}
  {021033} (\bibinfo {year} {2018})}\BibitemShut {NoStop}%
\bibitem [{\citenamefont {Bu}\ \emph {et~al.}(2017)\citenamefont {Bu},
  \citenamefont {Singh}, \citenamefont {Fei}, \citenamefont {Pati},\ and\
  \citenamefont {Wu}}]{Bu2017}%
  \BibitemOpen
  \bibfield  {author} {\bibinfo {author} {\bibfnamefont {K.}~\bibnamefont
  {Bu}}, \bibinfo {author} {\bibfnamefont {U.}~\bibnamefont {Singh}}, \bibinfo
  {author} {\bibfnamefont {S.-M.}\ \bibnamefont {Fei}}, \bibinfo {author}
  {\bibfnamefont {A.~K.}\ \bibnamefont {Pati}}, \ and\ \bibinfo {author}
  {\bibfnamefont {J.}~\bibnamefont {Wu}},\ }\href {\doibase
  10.1103/PhysRevLett.119.150405} {\bibfield  {journal} {\bibinfo  {journal}
  {Phys. Rev. Lett.}\ }\textbf {\bibinfo {volume} {119}},\ \bibinfo {pages}
  {150405} (\bibinfo {year} {2017})}\BibitemShut {NoStop}%
\bibitem [{\citenamefont {Regula}\ \emph {et~al.}(2020)\citenamefont {Regula},
  \citenamefont {Lami}, \citenamefont {Ferrari},\ and\ \citenamefont
  {Takagi}}]{regula20}%
  \BibitemOpen
  \bibfield  {author} {\bibinfo {author} {\bibfnamefont {B.}~\bibnamefont
  {Regula}}, \bibinfo {author} {\bibfnamefont {L.}~\bibnamefont {Lami}},
  \bibinfo {author} {\bibfnamefont {G.}~\bibnamefont {Ferrari}}, \ and\
  \bibinfo {author} {\bibfnamefont {R.}~\bibnamefont {Takagi}},\ }\href@noop {}
  {\enquote {\bibinfo {title} {Operational quantification of
  continuous-variable quantum resources},}\ } (\bibinfo {year} {2020}),\
  \Eprint {http://arxiv.org/abs/2009.11302} {arXiv:2009.11302 [quant-ph]}
  \BibitemShut {NoStop}%
\bibitem [{\citenamefont {Lami}\ \emph {et~al.}(2020)\citenamefont {Lami},
  \citenamefont {Regula}, \citenamefont {Takagi},\ and\ \citenamefont
  {Ferrari}}]{lami20}%
  \BibitemOpen
  \bibfield  {author} {\bibinfo {author} {\bibfnamefont {L.}~\bibnamefont
  {Lami}}, \bibinfo {author} {\bibfnamefont {B.}~\bibnamefont {Regula}},
  \bibinfo {author} {\bibfnamefont {R.}~\bibnamefont {Takagi}}, \ and\ \bibinfo
  {author} {\bibfnamefont {G.}~\bibnamefont {Ferrari}},\ }\href@noop {}
  {\enquote {\bibinfo {title} {Taming the infinite: framework for resource
  quantification in infinite-dimensional general probabilistic theories},}\ }
  (\bibinfo {year} {2020}),\ \Eprint {http://arxiv.org/abs/2009.11313}
  {arXiv:2009.11313 [quant-ph]} \BibitemShut {NoStop}%
\bibitem [{\citenamefont {Haapasalo}(2015)}]{Haapasalo2015}%
  \BibitemOpen
  \bibfield  {author} {\bibinfo {author} {\bibfnamefont {E.}~\bibnamefont
  {Haapasalo}},\ }\href {\doibase 10.1088/1751-8113/48/25/255303} {\bibfield
  {journal} {\bibinfo  {journal} {J. Phys. A}\ }\textbf {\bibinfo {volume}
  {48}},\ \bibinfo {pages} {255303} (\bibinfo {year} {2015})}\BibitemShut
  {NoStop}%
\bibitem [{\citenamefont {Davies}(1969)}]{davies1969}%
  \BibitemOpen
  \bibfield  {author} {\bibinfo {author} {\bibfnamefont {E.~B.}\ \bibnamefont
  {Davies}},\ }\href {https://projecteuclid.org:443/euclid.cmp/1103841988}
  {\bibfield  {journal} {\bibinfo  {journal} {Comm. Math. Phys.}\ }\textbf
  {\bibinfo {volume} {15}},\ \bibinfo {pages} {277} (\bibinfo {year}
  {1969})}\BibitemShut {NoStop}%
\bibitem [{\citenamefont {Streltsov}\ \emph {et~al.}(2017)\citenamefont
  {Streltsov}, \citenamefont {Adesso},\ and\ \citenamefont
  {Plenio}}]{Streltsov2017}%
  \BibitemOpen
  \bibfield  {author} {\bibinfo {author} {\bibfnamefont {A.}~\bibnamefont
  {Streltsov}}, \bibinfo {author} {\bibfnamefont {G.}~\bibnamefont {Adesso}}, \
  and\ \bibinfo {author} {\bibfnamefont {M.~B.}\ \bibnamefont {Plenio}},\
  }\href {\doibase 10.1103/RevModPhys.89.041003} {\bibfield  {journal}
  {\bibinfo  {journal} {Rev. Mod. Phys.}\ }\textbf {\bibinfo {volume} {89}},\
  \bibinfo {pages} {041003} (\bibinfo {year} {2017})}\BibitemShut {NoStop}%
\bibitem [{\citenamefont {Piani}\ \emph {et~al.}(2016)\citenamefont {Piani},
  \citenamefont {Cianciaruso}, \citenamefont {Bromley}, \citenamefont {Napoli},
  \citenamefont {Johnston},\ and\ \citenamefont {Adesso}}]{Piani2016}%
  \BibitemOpen
  \bibfield  {author} {\bibinfo {author} {\bibfnamefont {M.}~\bibnamefont
  {Piani}}, \bibinfo {author} {\bibfnamefont {M.}~\bibnamefont {Cianciaruso}},
  \bibinfo {author} {\bibfnamefont {T.~R.}\ \bibnamefont {Bromley}}, \bibinfo
  {author} {\bibfnamefont {C.}~\bibnamefont {Napoli}}, \bibinfo {author}
  {\bibfnamefont {N.}~\bibnamefont {Johnston}}, \ and\ \bibinfo {author}
  {\bibfnamefont {G.}~\bibnamefont {Adesso}},\ }\href {\doibase
  10.1103/PhysRevA.93.042107} {\bibfield  {journal} {\bibinfo  {journal} {Phys.
  Rev. A}\ }\textbf {\bibinfo {volume} {93}},\ \bibinfo {pages} {042107}
  (\bibinfo {year} {2016})}\BibitemShut {NoStop}%
\bibitem [{\citenamefont {Holevo}\ \emph {et~al.}(2005)\citenamefont {Holevo},
  \citenamefont {Shirokov},\ and\ \citenamefont {Werner}}]{HoShiWe2005}%
  \BibitemOpen
  \bibfield  {author} {\bibinfo {author} {\bibfnamefont {A.}~\bibnamefont
  {Holevo}}, \bibinfo {author} {\bibfnamefont {M.}~\bibnamefont {Shirokov}}, \
  and\ \bibinfo {author} {\bibfnamefont {R.}~\bibnamefont {Werner}},\
  }\href@noop {} {\bibfield  {journal} {\bibinfo  {journal} {Russ. Math.
  Surv.}\ }\textbf {\bibinfo {volume} {60}} (\bibinfo {year}
  {2005})}\BibitemShut {NoStop}%
\bibitem [{\citenamefont {Horodecki}\ \emph {et~al.}(2009)\citenamefont
  {Horodecki}, \citenamefont {Horodecki}, \citenamefont {Horodecki},\ and\
  \citenamefont {Horodecki}}]{Horodecki2009}%
  \BibitemOpen
  \bibfield  {author} {\bibinfo {author} {\bibfnamefont {R.}~\bibnamefont
  {Horodecki}}, \bibinfo {author} {\bibfnamefont {P.}~\bibnamefont
  {Horodecki}}, \bibinfo {author} {\bibfnamefont {M.}~\bibnamefont
  {Horodecki}}, \ and\ \bibinfo {author} {\bibfnamefont {K.}~\bibnamefont
  {Horodecki}},\ }\href {\doibase 10.1103/RevModPhys.81.865} {\bibfield
  {journal} {\bibinfo  {journal} {Rev. Mod. Phys.}\ }\textbf {\bibinfo {volume}
  {81}},\ \bibinfo {pages} {865} (\bibinfo {year} {2009})}\BibitemShut
  {NoStop}%
\bibitem [{\citenamefont {Gühne}\ and\ \citenamefont
  {Tóth}(2009)}]{Guehne2009}%
  \BibitemOpen
  \bibfield  {author} {\bibinfo {author} {\bibfnamefont {O.}~\bibnamefont
  {Gühne}}\ and\ \bibinfo {author} {\bibfnamefont {G.}~\bibnamefont {Tóth}},\
  }\href {\doibase https://doi.org/10.1016/j.physrep.2009.02.004} {\bibfield
  {journal} {\bibinfo  {journal} {Phys. Rep.}\ }\textbf {\bibinfo {volume}
  {474}},\ \bibinfo {pages} {1 } (\bibinfo {year} {2009})}\BibitemShut
  {NoStop}%
\bibitem [{\citenamefont {Heinosaari}\ and\ \citenamefont
  {Miyadera}(2017)}]{Heinosaari2017}%
  \BibitemOpen
  \bibfield  {author} {\bibinfo {author} {\bibfnamefont {T.}~\bibnamefont
  {Heinosaari}}\ and\ \bibinfo {author} {\bibfnamefont {T.}~\bibnamefont
  {Miyadera}},\ }\href {\doibase 10.1088/1751-8121/aa5f6b} {\bibfield
  {journal} {\bibinfo  {journal} {Journal of Physics A: Mathematical and
  Theoretical}\ }\textbf {\bibinfo {volume} {50}},\ \bibinfo {pages} {135302}
  (\bibinfo {year} {2017})}\BibitemShut {NoStop}%
\bibitem [{\citenamefont {Klyachko}(2006)}]{Klyachko2006}%
  \BibitemOpen
  \bibfield  {author} {\bibinfo {author} {\bibfnamefont {A.~A.}\ \bibnamefont
  {Klyachko}},\ }\href {\doibase 10.1088/1742-6596/36/1/014} {\bibfield
  {journal} {\bibinfo  {journal} {Journal of Physics: Conference Series}\
  }\textbf {\bibinfo {volume} {36}},\ \bibinfo {pages} {72–86} (\bibinfo
  {year} {2006})}\BibitemShut {NoStop}%
\end{thebibliography}%

\pagebreak

\appendix

\onecolumngrid

\section{Appendix A: The connection between multi-channel games and resource quantifiers in the finite-dimensional case}
\label{app:tuples}

For sets of channels, the needed modification to quantum games is that one considers individual games for each channel, and the payoff is defined as the sum of individual payoffs. Formally, we have a game $\vec{\mc G}$ consisting of individual games $\mc G_i$ having individual input states $\{\varrho^i_a\}_a$, measurements $\{M^i_b\}_b$, and score assignment $\{\omega_{ab}^i\}_{ab}$. For a channel tuple $\vec\Lambda=\{\Lambda_i\}_i$ the payoff is defined as
\begin{align}
    \mathcal P(\vec\Lambda,\vec{\mc G}):=\sum_{a,b,i}\omega_{ab}^i\text{tr}[\Lambda_i(\varrho^i_a)M_b^i].
\end{align}
It is straight-forward to check that the connections given by Eq.~(\ref{eqn:ChanRobAdv}) and Eq.~(\ref{eqn:ChanWeiAdv}) work for state tuples as well

For a channel tuple $\vec\Lambda$, the robustness measure is defined as
\begin{equation} \label{eqn:RobustnessChantuple}
  \mathcal R_{F,N}(\vec\Lambda) = \min\qty{t\geq 0\,\,\bigg|\,\, \frac{\vec\Lambda + t\vec{\tilde\Lambda}}{1+t}\in F,\,\vec{\tilde\Lambda}\in N},
\end{equation}
where $\vec{\tilde\Lambda}$ is a channel tuple of the same length as $\vec\Lambda$, and the set $F$ is a convex and closed set of channel tuples of the same length. The convex weight is defined as
\begin{equation} \label{eqn:WeightChantuple}
  \mathcal W_F(\vec\Lambda) = \min\qty{\mu\geq 0\,\,\bigg|\,\,\vec\Lambda= \mu\vec{\tilde\Lambda}+(1-\mu)\vec\Lambda_F},
\end{equation}
where $\vec{\tilde\Lambda}$ is a channel tuple with the same length as $\vec\Lambda$, and $\vec\Lambda_F\in F$. Using again the techniques of Refs.~\cite{TakagiRegula19,Uola20,UolaBullock20}, one sees that the duals of these quantifiers read
\begin{eqnarray}
\label{eqn:RobDualtuple}
1+\mathcal R_{F,N}(\vec\Lambda) = \max_{\vec Y}\, && \sum_i\text{tr}[\Lambda_i Y_i] \\
\text{s.t.: }&& \sum_i\text{tr}[\Lambda_F^i Y_i]\leq 1\, \forall\ \vec\Lambda_F\in F \notag\\
&& \sum_i\text{tr}[\tilde\Lambda_i Y_i]\geq0\, \forall\ \vec{\tilde\Lambda}\in N \notag
\end{eqnarray}
and
\begin{eqnarray}
\label{eqn:WeiDualtuple}
    1-\mathcal{W}_F(\vec\Lambda) = \min_{\vec Y\geq0}\, && \sum_i\text{tr}[\Lambda_i Y_i] \\
\text{s.t.: }&& \sum_i\text{tr}[\Lambda_F^i Y_i]\geq1\ \forall \ \vec\Lambda_F\in F\notag,
\end{eqnarray}
where $\vec Y=\{Y_i\}_i$.

To see the connection to quantum games $\vec{\mc G}$, we note that the components of the witness $\vec Y$ can be written again as $Y_i=\sum_{ab}\omega_{ab}^i(\varrho^i_a)^T\otimes M^i_b$, which shows that for an optimal witness $\vec Y^r$ of the robustness and $\vec Y^w$ of the weight, the corresponding instances of the games $\vec{\mc G}_{\vec Y^r}$ and $\vec{\mc G}_{\vec Y^w}$ give
\begin{align}
\label{eqn:RobineqChantuple}
    \frac{\mathcal P(\vec\Lambda,\vec{\mc G}_{\vec Y^r})}{\max_{\vec\Lambda_F\in F} \mathcal P(\vec\Lambda_F,\vec{\mc G}_{Y^r})}\geq 1+ \mathcal R_{F,N}(\vec\Lambda)
\end{align}
and
\begin{align}
\label{eqn:WeiineqChantuple}
    \frac{\mathcal P(\vec\Lambda,\vec{\mc G}_{Y^w})}{\min_{\vec\Lambda_F\in F} \mathcal P(\vec\Lambda_F,\vec{\mc G}_{Y^w})}\leq 1- \mathcal W_F(\vec\Lambda).
\end{align}
Using the form of $\Lambda$ from Eq.~(\ref{eqn:RobustnessChantuple}) and Eq.~(\ref{eqn:WeightChantuple}) together with the linearity of the games, we have
\begin{align}
\label{eqn:RobAdvChantuple}
    \sup_{\vec{\mc G}_N}\frac{\mathcal P(\vec\Lambda,\vec{\mc G}_N)}{\max_{\vec\Lambda_F\in F} \mathcal P(\vec\Lambda_F,\vec{\mc G}_N)}= 1+ \mathcal R_{F,N}(\vec\Lambda)
\end{align}
and
\begin{align}
\label{eqn:WeiAdvChantuple}
    \inf_{\vec{\mc G}}\frac{\mathcal P(\vec\Lambda,\vec{\mc G})}{\min_{\vec\Lambda_F\in F} \mathcal P(\vec\Lambda_F,\vec{\mc G})}= 1- \mathcal W_F(\vec\Lambda),
\end{align}
where the optimization goes as follows: for $N$ being the whole set, as well as for the convex weight, the optimization runs over those games $\vec{\mc G}_N$ that have a non-negative payoff for any input channel tuple, and for the case $N=F$, one optimises over those games $\vec{\mc G}_N$ that have a non-negative payoff for the free set, and in every case we leave out games that would make the corresponding denominator zero.

\section{Appendix B: Properties of the approximate robustness}
\label{app:Properties}

Let us fix an approximation $\mathbb A=\{\alpha_n,\beta_n\}_n$. We first show that the sequence $\big(\mc R_{\overline{F}_n,\overline{N}_n}(\beta_n\circ\Lambda\circ\alpha_n)\big)_n$ is increasing. Using our final assumption when defining the approximations, it follows that, whenever $m\leq n$, $\{\beta_m\circ\Lambda_F\circ\alpha_m\,|\,\Lambda_F\in F_n\}=F_m$ and similarly for $N_m$ and $N_n$. Moreover, whenever $m\leq n$, we have $\mc R_{\overline{F}_m,\overline{N}_m}(\beta_m\circ\Lambda\circ\alpha_m)\leq\mc R_{\overline{F}_n,\overline{N}_n}(\beta_n\circ\Lambda\circ\alpha_n)$ for any channel $\Lambda$; note that the closures $\overline{F}_n$ and $\overline{N}_n$ are taken within the finite-dimensional set of channels between $\mc H_n$ and $\mc K_n$. Indeed, if $R\geq0$ and $\Lambda_N\in N_n$ are such that $(1+R)^{-1}(\beta_n\circ\Lambda\circ\alpha_n+R\Lambda_N)\in F_n$ (so that $R\geq\mc R_{\overline{F}_n,\overline{N}_n}(\beta_n\circ\Lambda\circ\alpha_n)$), we have
$$
\frac{1}{1+R}(\beta_m\circ\beta_n\circ\Lambda\circ\alpha_n\circ\alpha_m+\beta_m\circ R\Lambda_N\circ\alpha_m)=\frac{1}{1+R}(\beta_m\circ\Lambda\circ\alpha_m+R\beta_m\circ \Lambda_N\circ\alpha_m)\in F_m,
$$
implying that $R\geq\mc R_{\overline{F}_m,\overline{N}_m}(\beta_m\circ\Lambda\circ\alpha_m)$. The claim now follows as we let $R\searrow\mc R_{\overline{F}_n,\overline{N}_n}(\beta_n\circ\Lambda\circ\alpha_n)$. Similarly, $\mc R_{\overline{F}_n,\overline{N}_n}(\beta_n\circ\Lambda\circ\alpha_n)\leq\mc R_{F,N}(\Lambda)$ for all channels $\Lambda$ and $n\in\N$. All in all, $\big(\mc R_{\overline{F}_n,\overline{N}_n}(\beta_n\circ\Lambda\circ\alpha_n)\big)_{n=1}^\infty$ is an increasing sequence upper-bounded by $\mc R_{F,N}(\Lambda)$.

Let us now prove the claim concerning the monotonicity of the approximate robustness $\mc R^{\mathbb A}_{F,N}$. Let $\Theta$ be operation (`superchannel') such that $\Theta(F)\subseteq F$ and $\Theta(N)\subseteq N$ and let $n\in\N$, $R\geq0$ and $\Lambda^N_n\in N_n$ be such that $(1+R)^{-1}(\beta_n\circ\Lambda\circ\alpha_n+\Lambda^N_n)=:\Lambda^F_n\in F_n$ and pick $\Lambda_N\in N$ such that $\beta_n\circ\Lambda_N\circ\alpha_n=\Lambda^N_n$ and $\Lambda_F\in F$ such that $\beta_n\circ\Lambda_F\circ\alpha_n=\Lambda^F_n$. We now have $(1+R)^{-1}(\beta_n\circ\Theta(\Lambda)\circ\alpha_n+R\beta_n\circ\Theta(\Lambda_N)\circ\alpha_n)=\beta_n\circ\Theta(\Lambda_F)\circ\alpha_n\in F_n$, so that
$$
\mc R_{\overline{F}_n,\overline{N}_n}\big(\beta_n\circ\Theta(\Lambda)\circ\alpha_n\big)\leq\mc R_{\overline{F}_n,\overline{N}_n}(\beta_n\circ\Lambda\circ\alpha_n)
$$
as we let $R\searrow\mc R_{\overline{F}_n,\overline{N}_n}(\beta_n\circ\Lambda\circ\alpha_n)$. Naturally, the same also holds for the supremums over $n$ of the LHS and RHS, implying the claim.

\section{Appendix C: Proof of Theorem \ref{Theorem1}}
\label{app:Proof1}

In this appendix, we first prove Theorem \ref{Theorem1}, and then give some examples where the conditions are met. We first present a useful lemma. We let $\mc V$ be a (real or complex) vector space and $A\subseteq\mc V$ be an affine plane, i.e.,\ whenever $x,\,y\in A$ and $\lambda\in\R$, then $\lambda x+(1-\lambda)y\in A$. We let $F\subseteq A$ be convex and $N$ be some other subset of $A$. For $x,\,y\in A$, we define $\mc R_F(x|y)\in[0,\infty]$ as the infimum of those $r\geq0$ such that $(1+r)^{-1}(x+ry)\in F$ where $\inf\emptyset:=\infty$. Moreover, we define the $(F,N)$-robustness, as usual through $\mc R_{F,N}(x)=\inf_{y\in N}\mc R_F(x|y)$.

\begin{lemma}\label{lemma:genRob}
The robustness function $\mc R_{F,N}:A\to[0,\infty]$ is convex. Whenever $\mc V$ is a topological vector space, $A$ and $F$ are closed, and $N$ is relatively compact, $\mc R_{F,N}$ is faithful w.r.t.\ $F$. If, additionally, $N$ is compact, $\mc R_{F,N}$ is lower semi-continuous.
\end{lemma}

\begin{proof}
The first claim concerning convexity is well known and proven, e.g.,\ in \cite{Haapasalo2015}. Let us assume that $\mc V$ is a topological vector space, $A$ and $F$ are closed, and $N$ is relatively compact. Showing that $\mc R_{F,N}(x)=0$ whenever $x\in F$ is easy; let us concentrate on showing the converse. Assume that $x\in A$ is such that $\mc R_{F,N}(x)=0$. It follows that, for all $n\in\N$, there are $r_n\in[0,1/n)$ and $y_n\in N$ such that $z_n:=(1+r_n)^{-1}(x+r_n y_n)\in F$. Using again the compactness of $N$ and, if necessary, by moving on to a subsequence, we find $y\in\overline{N}$ (where $\overline{N}$ is the compact closure of $N$) such that $y_n\overset{n\to\infty}{\to}y$. As $\lim_{n\to\infty}r_n=0$, we see that $z_n\overset{n\to\infty}{\to}(1+0)^{-1}(x+0y)=x$ and, since $F$ is closed, we find $x\in F$.

Let us, additionally, assume that $N$ is compact. We make the counter assumption that $\mc R_{F,N}$ is not lower semi-continuous. This means that there is $\alpha\geq0$, $x\in A$, and a net $(x_\lambda)_{\lambda\in L}\subset A$ which converges to $x$ such that $\alpha<\mc R_{F,N}(x)$ and $\mc R_{F,N}(x_\lambda)\leq\alpha$ for all $\lambda\in L$. Denote $\varepsilon:=(1/2)\big(\mc R_{F,N}(x)-\alpha\big)>0$. It follows that, for all $\lambda\in L$, there are $r_\lambda\in[0,\alpha+\varepsilon]$ and $y_\lambda\in N$ such that $z_\lambda:=(1+r_\lambda)^{-1}(x_\lambda+r_\lambda y_\lambda)\in F$ (i.e.,\ $\mc R_{F,N}(x_\lambda)\leq r_\lambda$). Using the compactness of $[0,\alpha+\varepsilon]$ and $N$ and, if necessary, by moving on to a subnet, we find $r\in[0,\alpha+\varepsilon]$ and $y\in N$ such that $r_\lambda\overset{\lambda\in L}{\to}r$ and $y_\lambda\overset{\lambda\in L}{\to}y$. As $x$ is the limit of $(x_\lambda)_{\lambda\in L}$, we now see that $z_\lambda\overset{\lambda\in L}{\to}z:=(1+r)^{-1}(x+ry)$ and, as $F$ is closed, $z\in F$. Thus,
$$
\mc R_{F,N}(x)\leq r\leq\alpha+\varepsilon<\mc R_{F,N}(x),
$$
a contradiction. Hence, $\mc R_{F,N}$ is lower semi-continuous.
\end{proof}

Let us resume the assumptions made before the statement of Theorem \ref{Theorem1}. Let us formally define the topology $\tau$: this is the coarsest topology within the set $\mc L(\mc K)^{\mc L(\hil)}$ of all functions $\Phi:\mc L(\mc K)\to\mc L(\hil)$ (where, e.g.,\ $\mc L(\hil)$ stands for the algebra of bounded linear operators on $\hil$) such that, for any $\varrho\in\mc T(\hil)$ (a trace-class operator) and $B\in\mc L(\mc K)$, the map $\mc L(\mc K)^{\mc L(\hil)}\ni\Phi\mapsto{\rm tr}[\varrho\Phi(B)]$ is continuous. The $\tau$-closure of the set ${\bf Ch}(\hil,\mc K)$ of all channels (i.e.\ completely positive trace-preserving maps) $\Lambda:\mc T(\hil)\to\mc T(\mc K)$ is the set ${\bf UCP}(\hil,\mc K)$ of unital completely positive linear maps $\Phi:\mc L(\mc K)\to\mc L(\hil)$. This set is, in fact, $\tau$-compact which is easily seen by using the Alao\u{g}lu-Bourbaki and Tykhonoff theorems.
Note that we identify any (Schr\"odinger) channel $\Lambda:\, \mathcal{T(H)}\to\mathcal{T(K)}$ with its dual (Heisenberg) channel $\Lambda^*:\, \mathcal{L(K)}\to\mathcal{L(H)}$ (defined by $\tr{\Lambda^*(B)\varrho}=\tr{B\Lambda(\varrho)}$).

We define $\overline{\mc R}_{F,N}:{\bf UCP}(\hil,\mc K)\to[0,\infty]$ by setting $\overline{\mc R}_{F,N}=\mc R_{\overline{F},\overline{N}}$ where $\overline{F}$ is the $\tau$-closure of $F$ and $\overline{N}$ is the $\tau$-closure of $N$. According to Lemma \ref{lemma:genRob}, $\overline{R}_{F,N}$ is lower semi-continuous and faithful w.r.t.\ $\overline{F}$. It clearly follows that $\overline{\mc R}_{F,N}(\Phi)\leq\mc R_{F,N}(\Phi)$ for any ${\bf UCP}(\hil,\mc K)$ and, especially, for any channel $\Phi$. We define the infinite-dimensional games in the obvious way: a game is a collection $\mc G=\big(\{\varrho_a\}_{a\in A},\{M_b\}_{b\in B},\{\omega_{a,b}\}_{a\in A,\,b\in B}\big)$ where $A$ and $B$ are finite sets, $\varrho_a$ is an input state on $\hil$ for all $a\in A$, $\{M_b\}_{b\in B}$ is a POVM, i.e.,\ a collection of positive operators on $\mc K$ summing up to the identity, and $\omega_{a,b}$, $a\in A$, $b\in B$, are real numbers, and the pay-off of $\mc G$ for a given channel $\Lambda$ is
$$
\mc P(\Lambda,\mc G):=\sum_{a\in A}\sum_{b\in B}\omega_{a,b}\tr{\Lambda(\varrho_a)M_b}.
$$
Note that, in principle the infinite-dimensional setting would allow for continuous measurements and state ensembles. Although our results would also work for a larger class of games arising from this setup, we consider only discrete games, as the desired suprema and infima can be arbitrarily well approximated by finite inputs and outputs. For any game $\mc G$, the map $\mc P(\cdot,\mc G)$ can be extended into a $\tau$-continuous map on ${\bf UCP}(\hil,\mc K)$; note that, in this setting, we have to write the pay-off in Heisenberg picture. Pick $\Phi\in{\bf UCP}(\hil,\mc K)$ and suppose, for a while, that $N=F$ or $N={\bf Ch}(\hil,\mc K)$. Suppose that $R\geq0$ and $\Phi_N\in\overline{N}$ are such that $(1+R)^{-1}(\Phi+R\Phi_N)=\Phi_F\in\overline{F}$. We have, for any game $\mc G$ such that $\mc P(\Phi_N,\mc G)\geq0$ for all $\Phi_N\in\overline{N}$, $(1+R)\mc P(\Phi_F,\mc G)-\mc P(\Phi,\mc G)=R\mc P(\Phi_N,\mc G)\geq0$, implying that
$$
(1+R)\sup_{\Phi_F\in\overline{F}}\mc P(\Phi_F,\mc G)-\mc P(\Phi,\mc G)\geq0
$$
Since $\overline{F}$ is compact and the pay-off function is $\tau$-continuous, we have $\sup_{\Lambda_F\in F}\mc P(\Lambda_F,\mc G)=\sup_{\Phi_F\in\overline{F}}\mc P(\Phi_F,\mc G)=\max_{\Phi_F\in\overline{F}}\mc P(\Phi_F,\mc G)\in\R$. Hence, by assuming that $\Phi=\Lambda^*$ is a channel and letting $R\searrow\overline{\mc R}_{F,N}(\Lambda)$,
$$
\frac{\mc P(\Lambda,\mc G)}{\sup_{\Lambda_F\in F}\mc P(\Lambda_F,\mc G)}\leq 1+\overline{\mc R}_{F,N}(\Lambda)\leq 1+\mc R_{F,N}(\Lambda).
$$

The above regularized robustness $\overline{\mc R}_{F,N}$ and its connection to outperformance in games will be useful in the following proof where we resume the assumptions made in the claim of Theorem \ref{Theorem1}.
\vspace{15pt}

\noindent{\it Proof of Theorem \ref{Theorem1}.} From our assumption it follows that $F$ is $\tau$-compact in $\mc L(\mc K)^{\mc L(\hil)}$ (or, equivalently, within ${\bf UCP}(\hil,\mc K)$). Let us first show that, for any $\Lambda\in{\bf Ch}(\hil,\mc K)$ and any approximation scheme $\mathbb A=\{\alpha_n,\beta_n\}_n$, $\beta_n\circ\Lambda\circ\alpha_n\to\Lambda$ w.r.t.\ $\tau$. Pick $\Lambda\in{\bf Ch}(\hil,\mc K)$, an approximation scheme $\mathbb A=\{\alpha_n,\beta_n\}_n$, and $\varrho\in\mc T(\hil)$ and $B\in\mc L(\mc K)$. We have
\begin{align*}
\left|\tr{\big(\Lambda(\varrho)-(\beta_n\circ\Lambda\circ\alpha_n)(\varrho)\big)B}\right|&\leq\left|\tr{\Lambda\big(\varrho-\alpha_n(\varrho)\big)B}\right|+\left|\tr{\Lambda\big(\alpha_n(\varrho)\big)\big(B-\beta_n^*(B)\big)}\right|\\
\leq&\|B\|\|\varrho-\alpha_n(\varrho)\|_{\rm tr}+\left|\tr{\Lambda\big(\varrho-\alpha_n(\varrho)\big)\big(B-\beta_n^*(B)\big)}\right|+\left|\tr{\Big(\Lambda(\varrho)-\beta_n\big(\Lambda(\varrho)\big)\Big)B}\right|\\
\leq&3\|B\|\underbrace{\|\varrho-\alpha_n(\varrho)\|_{\rm tr}}_{\overset{n\to\infty}{\to}0}+\|B\|\underbrace{\left\|\Lambda(\varrho)-\beta_n\big(\Lambda(\varrho)\big)\right\|_{\rm tr}}_{\overset{n\to\infty}{\to}0}\overset{n\to\infty}{\to}0.
\end{align*}

Let us fix $\Lambda\in{\bf Ch}(\hil,\mc K)$ and an approximation scheme $\mathbb A=\{\alpha_n,\beta_n\}_n$. Denote $R_n:=\mc R_{\overline{F}_n,\overline{N}_n}(\beta_n\circ\Lambda\circ\alpha_n)$ and $\Lambda_n:=\beta_n\circ\Lambda\circ\alpha_n$ for all $n\in\N$. For any $n\in\N$, we find $\Lambda^N_n\in N_n$ such that $\Lambda^F_n:=(1+R_n)^{-1}(\Lambda_n+R_n\Lambda^N_n)\in F_n$; note that, as an image of the $\tau$-compact set $F$ in an obviously $\tau$-continuous map $\Lambda_F\mapsto\beta_n\circ\Lambda_F\circ\alpha_n$, $F_n$ is compact and, hence, $\overline{F}_n=F_n$. Possibly by again moving on to a subsequence, we may assume that $\tau-\lim_{n\to\infty}\Lambda^N_n=\Phi_N\in{\bf UCP}(\hil,\mc K)$. If $N={\bf Ch}(\hil,\mc K)$, the rest of this paragraph can be omitted. Assume now that $N$ is $\tau$-closed in $\mc L(\mc K)^{\mc L(\hil)}$. This means that $N$ is $\tau$-compact as a subset of the $\tau$-compact set ${\bf UCP}(\hil,\mc K)$. Pick, for any $n\in\N$, $\Gamma^N_n\in N$ such that $\beta_n\circ\Gamma^N_n\circ\alpha_n=\Lambda^N_n$. Using the $\tau$-compactness of $N$ and by possibly moving on to a subsequence, we have some $\Lambda_N\in N$ such that $\tau-\lim_{n\to\infty}\Gamma^N_n=\Lambda_N$. Let us show that $\Phi_N=\Lambda_N^*$. For this, let us fix $\varrho\in\mc T(\hil)$ and $B\in\mc L(\mc K)$. We have, for all $n\in\N$,
\begin{align*}
\left|\tr{\varrho\big(\Phi_N(B)-\Lambda_N^*(B)\big)}\right|\leq&\left|\tr{\varrho\big(\Phi_N(B)-\Lambda^{N\,*}_n(B)\big)}\right|+\left|\tr{\varrho\big(\Lambda_N^*(B)-(\alpha_n^*\circ\Lambda_N^*\circ\beta_n^*)(B)\big)}\right|\\
&+\left|\tr{\Big(\Lambda_N\big(\alpha_n(\varrho)\big)-\Gamma^N_n\big(\alpha_n(\varrho)\big)\Big)\beta_n^*(B)}\right|.
\end{align*}
Above, the first two terms in the upper bound are easily seen to converge to zero as $n\to\infty$, so let us concentrate on the last term:
\begin{align*}
\left|\tr{\Big(\Lambda_N\big(\alpha_n(\varrho)\big)-\Gamma^N_n\big(\alpha_n(\varrho)\big)\Big)\beta_n^*(B)}\right|\leq&\|B\|\left\|\Lambda_N\big(\alpha_n(\varrho)\big)-\Gamma^N_n\big(\alpha_n(\varrho)\big)\right\|_{\rm tr}\\
\leq&\|B\|\left(\|\Lambda_N(\varrho)-\Gamma^N_n(\varrho)\|_{\rm tr}+\left\|(\Lambda_N-\Gamma^N_n)\big(\varrho-\alpha_n(\varrho)\big)\right\|_{\rm tr}\right)\\
\leq&\|B\|\big(\|\Lambda_N(\varrho)-\Gamma^N_n(\varrho)\|_{\rm tr}+2\|\varrho-\alpha_n(\varrho)\|_{\rm tr}\big).
\end{align*}
Let us look at the two terms in parentheses after the final inequality above: We have that $\Gamma^N_n(\varrho)\to\Lambda_N(\varrho)$ as $n\to\infty$ w.r.t.\ the $\sigma$-weak topology generated by the dual $\mc L(\mc K)$ of $\mc T(\mc K)$. Utilizing Lemma 4.3 of \cite{davies1969}, we find that the first term in parentheses converges to 0. The second term clearly converges to 0 by the definition of approximations. Thus, $\Phi_N=\Lambda_N^*$.

As $R_n\overset{n\to\infty}{\to}R:=\mc R_{F,N}^{\mathbb A}(\Lambda)$, $\Lambda_n\overset{n\to\infty}{\to}\Lambda$ w.r.t.\ $\tau$ (see the first paragraph of this proof) and $\Lambda^N_n\overset{n\to\infty}{\to}\Phi_N$ (where $\Phi_N$ is in the $\tau$-closure of $N$; note that this closure might differ from $N$ only in the case when $N={\bf Ch}(\hil,\mc K)$), we have $\Lambda^{F\,*}_n\overset{n\to\infty}{\to}(1+R)^{-1}(\Lambda^*+R\Phi_N)=:\Phi_F\in{\bf UCP}(\hil,\mc K)$ w.r.t.\ $\tau$. By picking, for any $n\in\N$, $\Gamma^F_n\in F$ such that $\beta_n\circ\Gamma^F_n\circ\alpha_n=\Lambda^F_n$ and, by using the assumption that $F$ is $\tau$-compact and possibly moving on to a subsequence, we find $\Lambda_F\in F$ such that $\tau-\lim_{n\to\infty}\Gamma^F_n=\Lambda_F$. Similarly as above we can show that $\Phi_F=\Lambda_F^*$. Thus, $(1+R)^{-1}(\Lambda^*+R\Phi_N)=\Lambda_F^*\in F^*$ where $F^*$ is the set of Heisenberg duals of free channels. Thus, $\Phi_N\in N^*$ as well (which was already known in the case where $N$ is $\tau$-closed). All in all, there is $\Lambda_N\in N$ such that $(1+R)^{-1}(\Lambda+R\Lambda_N)\in F$, implying that $R\geq\mc R_{F,N}(\Lambda)$. As we have earlier seen that $R\leq\mc R_{F,N}(\Lambda)$, we have proven the first claim. Moreover, these robustness measures further coincide with $\mc R_{\overline{F},N}$ under the assumptions of the claim, and, according to Lemma \ref{lemma:genRob}, this measure is faithful w.r.t. $\overline{F}=F$, proving the second claim.

Let us keep the channel $\Lambda\in{\bf Ch}(\hil,\mc K)$ and the approximation process $\mathbb A=\{\alpha_n,\beta_n\}_n$ fixed and assume that $N$ is $F$ or ${\bf Ch}(\hil,\mc K)$. Just before this proof we have seen that, for any (infinite-dimensional) game $\mc G$ such that $\mc P(\Lambda_N,\mc G)\geq0$ for all $\Lambda_N\in N$,
$$
\frac{\mc P(\Lambda,\mc G)}{\sup_{\Lambda_F\in F}\mc P(\Lambda_F,\mc G)}\leq 1+\mc R_{F,N}(\Lambda).
$$
Let us fix $\varepsilon>0$ and show that there is a game $\mc G$ such that $\mc P(\Lambda_N,\mc G)\geq0$ for all $\Lambda_N\in N$ and the LHS of the above inequality is closer than $\varepsilon$ to the RHS. We still denote $\Lambda_n:=\beta_n\circ\Lambda\circ\alpha_n$ and $R_n:=\mc R_{\overline{F}_n,\overline{N}_n}(\Lambda_n)$ for all $n\in\N$. Using the established fact that $\mc R_{F,N}^{\mathbb A}(\Lambda)=\mc R_{F,N}(\Lambda)=:R$ and the definition of $\mc R_{F,N}^{\mathbb A}(\Lambda)$, we find $n_0\in\N$ such that $R_{n_0}+\varepsilon/2>R$. Suppose that $\Lambda_0\in F$ is such that $\Lambda_0(\varrho)$ is faithful for all input states $\varrho$. It easily follows that $\beta_n\circ\Lambda_0\circ\alpha_n$ takes any state into a full-rank state on $\mc K_n$. The Choi-state of $\beta_n\circ\Lambda_0\circ\alpha_n$ is now easily seen to be of full rank, so that the Slater condition holds for any approximation step $n\in\N$; note that the evaluation of $R_n$ can be viewed as a finite-dimensional problem. Thus, there is a game $\mc G_0=\big(\{\varrho^0_a\}_{a\in A},\{M^0_b\}_{b\in B},\{\omega^0_{a,b}\}_{a\in A,\, b\in B}\big)$ (where $\# A,\,\# B<\infty$) such that
$$
\frac{\mc P(\Lambda_{n_0},\mc G_0)}{\sup_{\Lambda^0_F\in F_{n_0}}\mc P(\Lambda^0_F,\mc G_0)}+\varepsilon/2>1+R_{n_0}.
$$
Let us define the new game $\mc G_1=\big(\{\varrho_a\}_{a\in A},\{M_b\}_{b\in B},\{\omega_{a,b}\}_{a\in A\,b\in B}\big)$ where $\varrho_a=\alpha_{n_0}(\varrho^0_a)$, $M_b=\beta_{n_0}^*(M^0_b)$, and $\omega_{a,b}=\omega^0_{a,b}$ for all $a\in A$ and $b\in B$; note that $\{M_b\}_{b\in B}$ is still a POVM since $\beta_{n_0}^*$ is unital. It easily follows that, for any $\Gamma\in{\bf Ch}(\hil,\mc K)$, $\mc P(\Gamma,\mc G_1)=\mc P(\beta_{n_0}\circ\Gamma\circ\alpha_{n_0},\mc G_0)$, so that $\mc P(\Lambda_N,\mc G_1)\geq0$ for all $\Lambda_N\in\N$ and (using the definition of $F_{n_0}$)
$$
\frac{\mc P(\Lambda_{n_0},\mc G_0)}{\sup_{\Lambda^0_F\in F_{n_0}}\mc P(\Lambda^0_F,\mc G_0)}=\frac{\mc P(\Lambda,\mc G_1)}{\sup_{\Lambda_F\in F}\mc P(\beta_{n_0}\circ\Lambda_F\circ\alpha_{n_0},\mc G_0)}=\frac{\mc P(\Lambda,\mc G_1)}{\sup_{\Lambda_F\in F}\mc P(\Lambda_F,\mc G_1)}.
$$
Hence,
\begin{align*}
1+R-\frac{\mc P(\Lambda,\mc G_1)}{\sup_{\Lambda_F\in F}\mc P(\Lambda_F,\mc G_1)}=&R-R_{n_0}+1+R_{n_0}-\frac{\mc P(\Lambda,\mc G_1)}{\sup_{\Lambda_F\in F}\mc P(\Lambda_F,\mc G_1)}\\
<&\varepsilon/2+1+R_{n_0}-\frac{\mc P(\Lambda_{n_0},\mc G_0)}{\sup_{\Lambda^0_F\in F_{n_0}}\mc P(\Lambda^0_F,\mc G_0)}<\varepsilon/2+\varepsilon/2=\varepsilon,
\end{align*}
proving Eq.~(\ref{eqn:ChanRobinfadv}). The final claim concerning the weight is obtained in essentially the same way with infimums conversed to supremums and other simple modifications. Since the proof is so similar, it is not given here but left for the reader to verify. \hfill $\square$
\vspace{15pt}

Let us take a look at a couple of examples where the conditions of Theorem \ref{Theorem1} hold. These examples involve states, i.e.,\ channels with the trivial input space so that approximation is only done on the output. Let us fix an orthonormal basis $\mc B=\{|k\>\}_k\subset\hil$. We say that a state $\varrho$ is $\mc B$-incoherent if it diagonalizes in $\mc B$ and denote the convex set of these states by $F_{\mc B-{\rm incoh}}$~\cite{Streltsov2017}. This set is, indeed, $\tau$-closed. To see this, let $(\varrho_\lambda)_{\lambda\in L}$ be a net in $F_{\mc B-{\rm incoh}}$ $\tau$-converging to some unital positive linear functional $f:\mc L(\hil)\to\C$. Let $k\neq\ell$. We have
$$
0=\<k|\varrho_\lambda|\ell\>\overset{\lambda\in L}{\to} f(|\ell\>\<k|),
$$
so that $f(|\ell\>\<k|)=0$. Thus, $f$ is effectively a positive linear functional on the Abelian algebra of operators diagonal in $\mc B$, i.e.,\ $f\in(\ell^\infty)^*=\ell^1$. It follows that, denoting $c_k=f(|k\>\<k|)=\lim_{\lambda\in L}\<k|\varrho_\lambda|k\>$ for all $k$,
$$
f(B)=\sum_{k}\<k|B|k\>c_k,\qquad B\in\mc L(\hil),
$$
so that $\sum_{k}c_k=1$ and $f$ corresponds to the $\mc B$-incoherent state $\varrho=\sum_k c_k|k\>\<k|$. Thus the robustness measure $\mc R_{F_{\mc B-{\rm incoh}},N}$ of $\mc B$-coherence has the game interpretation where $N$ is the whole set of states, and a similar argument holds for the convex weight. Note that the free robustness has only the values zero and infinity in the case of coherence.

A related example is provided by asymmetry w.r.t.\ a compact group. Let $G$ be a compact group and $U:G\to\mc U(\hil)$ (where $\mc U(\hil)$ is the group of unitary operators on $\hil$) be a strongly continuous unitary representation. We could make the weaker assumption where $U$ is only a projective representation, however we can even in this case transform $U$ into an ordinary representation by group extension methods. We say that a state $\varrho$ is $U$-symmetric if $\varrho U(g)=U(g)\varrho$ for all $g\in G$ and denote the set of these states by $F_{U-{\rm symm}}$~\cite{Piani2016}. We call $\varrho\notin F_{U-{\rm symm}}$ as $U$-asymmetric. We denote the countable set of unitary equivalence classes of irreducible strongly continuous unitary representations by $\hat{G}$. For each $[\gamma]\in\hat{G}$, we choose a representative $\gamma:G\to\mc U(\mc K_\gamma)$. Using the Peter-Weyl theorem, we find Hilbert spaces $\mc M_\gamma$ such that $\hil=\bigoplus_{[\gamma]\in\hat{G}}\mc K_\gamma\otimes\mc M_\gamma$ and $U(g)=\bigoplus_{[\gamma]\in\hat{G}}\gamma(g)\otimes\openone_{\mc M_\gamma}$. It follows that, denoting $d_\gamma={\rm dim}(\mc K_\gamma)<\infty$ for all $[\gamma]\in\hat{G}$, for any $\varrho\in F_{U-{\rm symm}}$ there are non-normalized positive trace-class operators $\sigma_\gamma$ on $\mc M_\gamma$ such that $\sum_{[\gamma]\in\hat{G}}\tr{\sigma_\gamma}=1$ and $\varrho=\bigoplus_{[\gamma]\in\hat{G}}d_\gamma^{-1}\openone_{\mc K_\gamma}\otimes\sigma_\gamma$. Using similar methods as above, we may show that, if the multiplicity spaces $\mc M_\gamma$ are finite-dimensional, then $F_{U-{\rm symm}}$ is $\tau$-closed.

\section{Appendix D: Proof of Observation \ref{Obs1}}
\label{app:Proof2}

Next we prove Observation \ref{Obs1} and then we give some examples where the conditions of this result are met. In the next proof, we resume these assumptions.
\vspace{15pt}

\noindent{\it Proof of Observation \ref{Obs1}.}
Let $\mathbb A=\{\alpha_n,\beta_n\}_n$ be an approximation procedure such that $F_n\subseteq F$ and $N_n\subseteq N$ for all $n\in\N$. Let us fix $\Lambda\in{\bf Ch}(\hil,\mc K)$ and define $\Lambda_n:=\beta_n\circ\Lambda\circ\alpha_n$ and $R_n:=\mc R_{\overline{F}_n,\overline{N}_n}(\Lambda_n)$. We first prove that $R_n=\overline{\mc R}_{F,N}(\Lambda_n)$ for all $n\in\N$ where the measure $\overline{\mc R}_{F,N}$ is the regularized robustness introduced in the preceding appendix. Pick $n\in\N$ and let $R\geq\overline{\mc R}_{F,N}(\Lambda)$ and $\Phi_N\in\overline{N}$ be such that $(1+R)^{-1}(\Lambda_n^*+\Phi_N)=\Phi_F\in\overline{F}$. By recalling that $\alpha_n$ and $\beta_n$ are idempotent, we have $(1+R)^{-1}(\Lambda_n^*+\alpha_n^*\circ\Phi_N\circ\beta_n^*)=\alpha_n^*\circ\Phi_F\circ\beta_n^*\in\overline{F}_n^*$ implying that $R\geq R_n$; note that $\alpha_n^*\circ\Phi_N\circ\beta_n^*\in\overline{N}_n^*$, as projecting with the approximation channels makes $\Phi_N$ into normal (i.e.,\ ultraweakly continuous), i.e.,\ a channel in its Heisenberg dual. As we let $R\searrow\overline{\mc R}_{F,N}(\Lambda_n)$, we obtain $\overline{\mc R}_{F,N}(\Lambda_n)\geq R_n$. Let us pick, on the other hand, $R'\geq R_n$ and $\Lambda^N_n\in N_n\subseteq N$ such that $(1+R')^{-1}(\Lambda_n+R'\Lambda^N_n)=:\Lambda^F_n\in F_n\subseteq F$. Hence we immediately see that $R'\geq\overline{\mc R}_{F,N}(\Lambda_n)$, so that, as we let $R'\searrow R_n$, we have $R_n\geq\overline{\mc R}_{F,N}(\Lambda_n)$. Thus, $R_n=\overline{\mc R}_{F,N}(\Lambda_n)$. We have established in the beginning of the proof of Theorem \ref{Theorem1} that $\Lambda_n\overset{n\to\infty}{\to}\Lambda$ w.r.t.\ $\tau$, and as $(R_n)_{n=1}^\infty=\big(\overline{\mc R}_{F,N}(\Lambda_n)\big)_{n=1}^\infty$ is an increasing sequence and $\overline{\mc R}_{F,N}$ is lower semi-continuous according to Lemma \ref{lemma:genRob} and the discussion thereafter, we have
$$
\mc R_{F,N}^{\mathbb A}(\Lambda)=\lim_{n\to\infty}R_n=\overline{\mc R}_{F,N}(\Lambda).
$$
This means that, for any approximation procedure $\mathbb A$ satisfying the condition of the claim, $\mc R_{F,N}^{\mathbb A}=\overline{\mc R}_{F,N}$, implying the independence claim. The game interpretation is obtained in exactly the same way as in the latter part of proof of Theorem \ref{Theorem1}. The claims concerning the convex weight are proven in exactly the same way with only minor tweaks.
\hfill $\square$
\vspace{15pt}

Let us take a look at a couple of examples. We first study the free set $F_{EB}$ consisting of entanglement-breaking channels $\Lambda_{M,\sigma}$, which are of the form~\cite{HoShiWe2005}
$$
\Lambda_{M,\sigma}(\varrho)=\int_\Omega\sigma(x)\,\tr{\varrho\,M(dx)},
$$
for a (possibly continuous) POVM $M:\Sigma\to\mc L(\hil)$ defined on a $\sigma$-algebra $\Sigma$ of a non-empty set $\Omega$ and a output-state-valued map $\sigma$ on $\Omega$ such that, for any $B\in\mc L(\mc K)$, the function $\tr{\sigma(\cdot)B}$ is Borel-measurable (w.r.t.\ the trace norm topology of the trace class). Let $N=F_{EB}$ or $N={\bf Ch}(\hil,\mc K)$. We now have, for any approximation procedure $\mathbb A=\{\alpha_n,\beta_n\}_n$, that $\beta_n\circ\Lambda_{M,\sigma}\circ\alpha_n\in F_{EB}$ whenever $\Lambda_{M,\sigma}\in F_{EB}$ and $n\in\N$. Indeed, for any $\Lambda_{M,\sigma}\in F_{EB}$ as above, for any input state $\varrho$, and $n\in\N$, we have
$$
(\beta_n\circ\Lambda_{M,\sigma}\circ\beta_n)(\varrho)=\int_\Omega\beta_n\big(\sigma(x)\big)\,\tr{\varrho\alpha_n^*\big(M(dx)\big)}=\Lambda_{\beta_n\circ\sigma,\alpha_n^*\circ M}(\varrho),
$$
implying this claim. Thus, the conditions of Observation \ref{Obs1} hold. However, the question whether the robustness $\overline{\mc R}_{F_{EB},N}$ associated with the games coincides with $\mc R_{F_{EB},N}$ is still left open; this hinges upon whether $F_{EB}$ is $\tau$-closed or not.

Also the resource of entanglement~\cite{Horodecki2009,Guehne2009} satisfies the conditions of Observation \ref{Obs1}. We denote by $F_{\rm sep}$ the trace-norm closure of the convex hull of product states $\varrho^1\otimes\cdots\otimes\varrho^p$ on $\hil^1\otimes\cdots\otimes\hil^p$. If $\varrho\in F_{\rm sep}$, the state $\varrho$ is fully separable and, otherwise, it is entangled~\cite{HoShiWe2005}. If we choose the approximation (note that only post-processing channels are used on states as their input space is trivial) $\{\beta_n\}_n=\{\beta^1_n\otimes\cdots\otimes\beta^p_n\}_n$ where, for all $i=1,\ldots,p$, $\{\beta^i_n\}_n$ is an approximation for states on $\hil^i$, we immediately see that $\beta_n(\varrho)\in F_{\rm sep}$ whenever $\varrho\in F_{\rm sep}$, implying that the conditions of Observation \ref{Obs1} indeed hold and the game interpretation for $\overline{\mc R}_{F_{\rm sep},N}$ exists whenever $N=F_{\rm sep}$ or $N$ is the whole set of states over $\hil^1\otimes\cdots\otimes\hil^p$.

Until now, we have formulated Theorem \ref{Theorem1} and Observation \ref{Obs1} in the framework of the resource potential of single channels. However, these results also hold for collections $\vec{\Lambda}=(\Lambda^i)_{i=1}^p$ ($p<\infty$) of channels $\Lambda^i:\mc T(\hil^i)\to\mc T(\mc K^i)$. In this setting, we have approximation procedures $\{\vec{\alpha}_n,\vec{\beta}_n\}_{n=1}^\infty$ where, e.g.,\ $\vec{\alpha}_n=(\alpha^1_n,\ldots,\alpha^p_n)$ and, for each $i=1,\ldots,p$, $\{\alpha^i_n,\beta^i_n\}_{n=1}^\infty$ is an approximation procedure for channels $\Lambda:\mc T(\hil^i)\to\mc T(\mc K^i)$. The topology $\tau$ is simply replaced by the $p$-fold product topology of the topologies corresponding to $\tau$ in each component. With these minor tweaks, the proofs of Theorem \ref{Theorem1} and Observation \ref{Obs1} can be extended to the resource theory of collections of channels. The games are now of the form $\vec{\mc G}=\big(\{\varrho^1_a\}_{a\in A_i},\{M^i_b\}_{b\in B_i},\{\omega^i_{a,b}\}_{a\in A_i,\,b\in B_i}\big)_{i=1}^p$ and
$$
\mc P(\vec{\Lambda},\vec{\mc G})=\sum_{i=1}^p\sum_{a\in A_i}\sum_{b\in B_i}\omega^i_{a,b}\tr{\Lambda^i(\varrho^i_a)M^i_b}.
$$

The resource theory of broadcastability, or more generally channel incompatibility~\cite{Heinosaari2017}, is an example of this multipartite scenario. We define $Z:=\prod_{i=1}^p{\bf Ch}(\hil,\mc K^i)$ and say that $\vec{\Lambda}=(\Lambda^i)\in Z$ is compatible if there is $\Gamma\in{\bf Ch}(\hil,\mc K^1\otimes\cdots\otimes\mc K^p)$ such that $\Lambda^i={\rm tr}_{\{\mc K^i\}^c}\circ\Gamma$ for all $i=1,\ldots,\,p$, i.e.,\ $\vec{\Lambda}$ can be seen as reduced dynamics of the joint or broadcasting channel $\Gamma$ (which is typically not unique). We define $F_{\rm comp}\subset Z$ as the set of compatible channel tuples and call $\vec{\Lambda}\in Z\setminus F_{\rm comp}$ as incompatible. In this setting, the relevant approximations are of the form $\{\alpha_n,\vec{\beta}_n\}_n$ where $\vec{\beta}_n=(\beta^1_n,\ldots,\beta^p_n)$, i.e.,\ there is a single sequence $(\alpha_n)_n$ for the shared input space $\hil$. Let us assume that $\vec{\Lambda}=(\Lambda^i)_{i=1}^p\in F_{\rm comp}$ and fix $n\in\N$. Supposing that $\Gamma$ is a joint channel for $\vec{\Lambda}$, it is easy to prove that $(\beta^1_n\otimes\cdots\otimes\beta^p_n)\circ\Gamma\circ\alpha_n$ is a joint channel for $\vec{\beta}_n\circ\vec{\Lambda}\circ\alpha_n:=(\beta^i_n\circ\Lambda^i\circ\alpha_n)_{i=1}^p$. Hence, $\vec{\beta}_n\circ\vec{\Lambda}\circ\alpha_n\in F$. It follows that, for $N=F_{\rm comp}$ or $N=Z$, the counterpart of the condition in Observation \ref{Obs1} holds so that the approximate robustness (w.r.t.\ any approximation) has the game interpretation.

Another example of the multipartite case is given by the state tuple resource related to the quantum marginal problem~\cite{Klyachko2006}. Here the free set $F_{\rm marg}$ is the set of state $p$-tuples $\vec{\varrho}=(\varrho^i)_{i=1}^p\in\prod_{i=1}^p\mc S(\hil\otimes\mc K^i)$ (where $\mc S(\hil\otimes\mc K^i)$ stands for the set of states over the tensor-product Hilbert space $\hil\otimes\mc K^i$) such that there is a state $\varrho$ over $\hil\otimes\mc K^1\otimes\cdots\otimes\mc K^p$ such that $\varrho^i={\rm tr}_{\{\hil\otimes\mc K^i\}^c}[\varrho]$ for all $i=1,\ldots,p$. We say that if $\vec{\varrho}\in F_{\rm marg}$, then the marginal problem associated to $\vec{\varrho}$ has a solution. Since the input space for states is trivial, the approximation procedures now only consist of channel tuples on the output, i.e.,\ $\mathbb A=\{\vec{\beta}_n\}_{n=1}^\infty$ where $\vec{\beta}_n=(\beta^i_n)_{i=1}^p$ and $\{\beta^i_n\}_n$ is an approximation for states in the space $\hil\otimes\mc K^i$. We are able to satisfy the conditions of Observation \ref{Obs1} if the above approximation is of the form where, for each $i=1,\ldots,p$ and $n\in\N$, $\beta^i_n=\gamma^0_n\otimes\gamma^i_n$ and $\{\gamma^0_n\}_n$ is some fixed approximation for states in $\hil$ and $\{\gamma^i_n\}_n$ is an approximation for states in $\mc K^i$ for $i=1,\ldots,p$. Indeed, the proof of this claim is simple and mirrors that of the same claim for the set $F_{\rm comp}$ of compatible channels. Thus, $\mc R^{\mathbb A}_{F_{\rm marg},N}=\overline{\mc R}_{F_{\rm marg},N}$ for any approximation like that above and, whenever $N=F_{\rm marg}$ or $N=\prod_{i=1}^p\mc S(\hil\otimes\mc K^i)$, this robustness has the game interpretation. Note that here the relevant games have the form $\vec{\mc G}=\big(\{M^i_b\}_{b\in B_i},\{\omega^i_{b}\}_{b\in B_i}\big)_{i=1}^p$, where $\{M^i_b\}_{b\in B_i}$ is a POVM and, for all $b\in B_i$, $\omega_b\in\R$ for all $i=1,\ldots,p$, (although other equivalet descriptions of these games exist with possible modification of the state tuples through quantum instruments) and the pay-off for a state tuple $\vec{\varrho}=(\varrho^i)_{i=1}^p$ reads
$$
\mc P(\vec{\varrho},\vec{\mc G})=\sum_{i=1}^p\sum_{b\in B_i}\omega^i_b\tr{\varrho^i M^i_b}.
$$

\end{document}